\documentclass[11pt]{article}
\usepackage[utf8]{inputenc}
\usepackage[left=1in,top=1in,right=1in,bottom=1in,head=.1in]{geometry}

\usepackage{mystyle}

\newcommand{\dom}{\mathrm{dom}}
\newcommand{\Prox}{\mathrm{prox}}
\newcommand{\Gap}{\mathrm{Gap}}

\usepackage{xspace}
\usepackage{xcolor}
\newcommand{\gray}[1]{{\color{gray}#1}}
\usepackage[font=small,labelfont=bf]{caption}
\usepackage{tabularx}
\newcolumntype{Y}{>{\centering\arraybackslash}X}
\newcolumntype{P}{>{\raggedleft\arraybackslash}X}
\usepackage{multirow}

\usepackage{algorithm}
\usepackage{algorithmic}

\newcommand{\beginsupplement}{
    \setcounter{section}{0}
    \renewcommand{\thesection}{S\arabic{section}}
    \setcounter{equation}{0}
    \renewcommand{\theequation}{S\arabic{equation}}
    \setcounter{table}{0}
    \renewcommand{\thetable}{S\arabic{table}}
    \setcounter{figure}{0}
    \renewcommand{\thefigure}{S\arabic{figure}}
    \newcounter{SIfig}
    \renewcommand{\theSIfig}{S\arabic{SIfig}}}

\usepackage{authblk}
\newcommand*\samethanks[1][\value{footnote}]{\footnotemark[#1]}

\title{\bf High Dimensional Portfolio Selection\\with Cardinality Constraints }
\author[$\dagger$]{Jin-Hong Du\thanks{equal contribution}}
\author[$\ddagger$]{Yifeng Guo\samethanks[1]}
\author[$\S$]{Xueqin Wang}
\affil[$\S$]{School of Management, University of Science and Technology of China}
\affil[$\ddagger$]{Department of Statistics and Actuarial Science, The University of Hong Kong}
\affil[$\dagger$]{Department of Statistics and Data Science, Carnegie Mellon University}

\begin{document}

\maketitle

\begin{abstract}

The expanding number of assets offers more opportunities for investors but poses new challenges for modern portfolio management (PM). As a central plank of PM, portfolio selection by expected utility maximization (EUM) faces uncontrollable estimation and optimization errors in ultrahigh-dimensional scenarios. Past strategies for high-dimensional PM mainly concern only large-cap companies and select many stocks, making PM impractical. We propose a sample-average approximation-based portfolio strategy to tackle the difficulties above with cardinality constraints. Our strategy bypasses the estimation of mean and covariance, the Chinese walls in high-dimensional scenarios. Empirical results on S\&P 500 and Russell 2000 show that an appropriate number of carefully chosen assets leads to better out-of-sample mean-variance efficiency. On Russell 2000, our best portfolio profits as much as the equally-weighted portfolio but reduces the maximum drawdown and the average number of assets by 10\% and 90\%,  respectively. 
The flexibility and the stability of incorporating factor signals for augmenting out-of-sample performances are also demonstrated.
Our strategy balances the trade-off among the return, the risk, and the number of assets with cardinality constraints. Therefore, we provide a theoretically sound and computationally efficient strategy to make PM practical in the growing global financial market.
\end{abstract}

\noindent%

\section{Introduction}

    In modern portfolio management (PM), the number of assets grows with economic development and globalization. Taking the US as an example, there are currently more than 3500 stocks with a capitalization of over 48 trillion dollars, not to mention the financial assets across different countries in global indices. Therefore, the investors face considerate challenges in building feasible strategies to gain substantial returns with controllable risks.
    Nevertheless, the traditional PM approaches cannot handle such high dimensional problems as the number of assets nowadays can be in the same order or even more extensive than the sample size, bringing both theoretical and empirical challenges to PM.

    As a widely used portfolio strategy, the mean-variance (MV) framework introduced by Markowitz in 1952 \citep{Mark1952} aims at approximating the expected utility maximization (EUM) problems, which laid the foundation of modern portfolio theory.
    The modern portfolio theory emphasizes the diversified portfolio construction with a trade-off between return and risk.
    Although MV strategy provides reasonable estimates of expected utility in his recent paper \citep{markowitz2014mean}, its limitations are also documented in lots of literature, see \cite{michaud1989markowitz,best1991sensitivity,kan2007optimal,ao2019approaching}.
    More specifically, as the number of assets increases, the resulting ``plug-in'' portfolio is susceptible to the estimation error and input uncertainty during the optimization process \citep{fabozzi2007robust} since the expected mean and covariance are hard to be estimated from samples in the high-dimensional scenarios. In contrast, directly maximizing expected utility needs the exact information of distribution which is more intractable in practice.

    In addition to the above theoretical challenges, the computational challenges hinder the exploration of large datasets.
    As the S\&P 500 index is the leading US equities benchmark and usually outperforms other indices, it has been used to test high-dimensional PM approaches, see, for example, \cite{fan2012vast,hautsch2019large,ao2019approaching,kremer2020sparse,ding2021high}. However, whether these approaches can effectively handle datasets with much more stocks remains unknown.
    Considering mid-and small-cap companies that may offer more investors more opportunities, \cite{pun2019linear} extend the analysis to an incomplete Russell 2000 index constituents where the stocks with volatility greater than 0.04 are discarded, involving up to 1420 stocks in which most existing PM approaches fail. The optimality of their methods can only be established when some technical assumptions are satisfied. Also, their methods suffer from the instability of covariance estimation, which leads to a high turnover rate and transaction fee in their cases.

    In reality, any rational investor may choose a reasonable number of assets to hold when considering the transaction fees, the management fees, the budget constraint, and even the mental cost of looking into too many assets \citep{gao2013optimal}.     
    The portfolio strategies must be feasible so the investors can efficiently operate and manage the portfolios.
    On the other hand, it is crucial to maintain portfolio diversification and reduce overall risk profiles by selecting a suitable number of assets.
    Such a number may depend on the correlations among their prices \citep{fieldsend2004cardinality}.
    Overall, it is reasonable to impose cardinality constraints on EUM.
    However, the high dimensionality of the assets also brings computational challenges of existing methods on solving cardinality-constrained problems, which are generally NP-hard because a combinatorial number of searches is required to obtain exact solutions \citep{zhang2019learning},
    though many relaxation methods are proposed to approximate it within a reasonable time \citep{chang2000heuristics,fieldsend2004cardinality, gao2013optimal}.

    As the global financial market gets expanding, efficient portfolio strategies for high-dimensional PM are of practical relevance.
    For example, the Russell 2000 dataset containing around 1900 stocks is more challenging and less predictable because of its inherently ultrahigh dimensionality, higher volatility and lower signal-to-noise ratio, which will be further investigated in Section \ref{sec:data}.
    Faced with those above theoretical and empirical challenges for the ultrahigh-dimensional PM, we aim to answer the following question:
    \begin{center}
        \emph{How can we design an efficient algorithm to deal with the trade-off among the return, the risk, and the number of holding assets for high-dimensional portfolio management?}
    \end{center}
    To this end, we design a general portfolio strategy induced from the sample average approximation (SAA)  of the EUM problem with cardinality and probability simplex constraints for a variety of risk-averse utility functions.
    Compared with MV-based approaches, our proposed strategy avoids the estimation error for estimating moments, making it applicable in high-dimensional scenarios. We convert the original problem to the $l_1$-regularized SAA problem with theoretical guarantees and propose an efficient proximal algorithm with acceleration by safe screening techniques. Our strategy is effective and efficient even when the number of assets is far more than the sample size. In the meantime, the investors easily identify just a few valuable assets that can form diversified portfolios. 
    Last but not least, our strategy highlights new perspectives that an appropriate number of carefully chosen assets leads to better out-of-sample mean-variance efficiency from extensive case studies.

\section{Data Characteristics}\label{sec:data}

    In this paper, we consider the following two well-known stock market indices. 
    (1) \textbf{S\&P 500}.
    As one of the most widely followed equity indices, the S\&P 500 measures the performance of 500 large companies that are publicly traded. 
    We collect daily data from the CRSP database, including the open prices, the close prices, and the (holding period) returns from 2010 to 2020.
    Our stock pool is formed by:
    (a) Initialization with the S\&P 500 index components listed from 2010 to 2020;
    (b) If mergers and acquisitions (M\&A) occur from 2010 to 2020, then the records before and after M\&A are concatenated to form only one stock, and the return on the date M\&A happened is filled with the growth rate on that date, which results in a pool of 714 stocks in total  ;
    (c) For each rebalancing date, we only consider stocks that were listed in the S\&P 500 index by that date and had complete records of historical data in the previous year.
    (2) \textbf{Russell 2000}. Rather than focusing on large-cap companies like the S\&P 500, the Russell 2000 index seeks to represent the overall performance of small- and mid-cap companies on the stock market in the US. Therefore, the smallest 2000 stocks in the Russell 3000 index are included to form the Russell 2000 index.  First, we select a group of 5555 Russell 2000 stocks listed from 2004 to 2020 following the same procedure mentioned above and collect daily stock data accordingly.
    \begin{figure}[!t]
        \centering
        \subfigure[{Heatmaps of correlations among stocks in S\&P 500 (left) and Russell 2000 (right).}]{
          \includegraphics[width=0.85\textwidth]{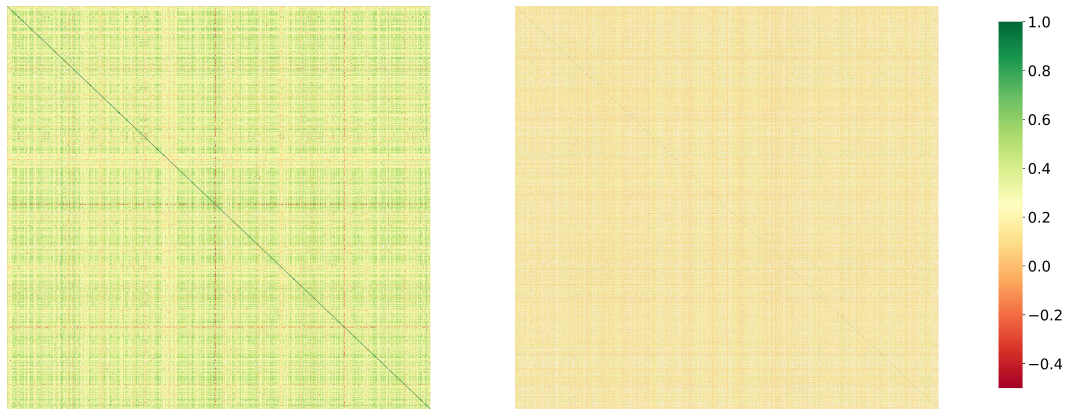} \label{subfig:cor}
        }
        \subfigure[ Proportional histogram of stock volatility in the two datasets.  ]{
          \includegraphics[width=0.85\textwidth]{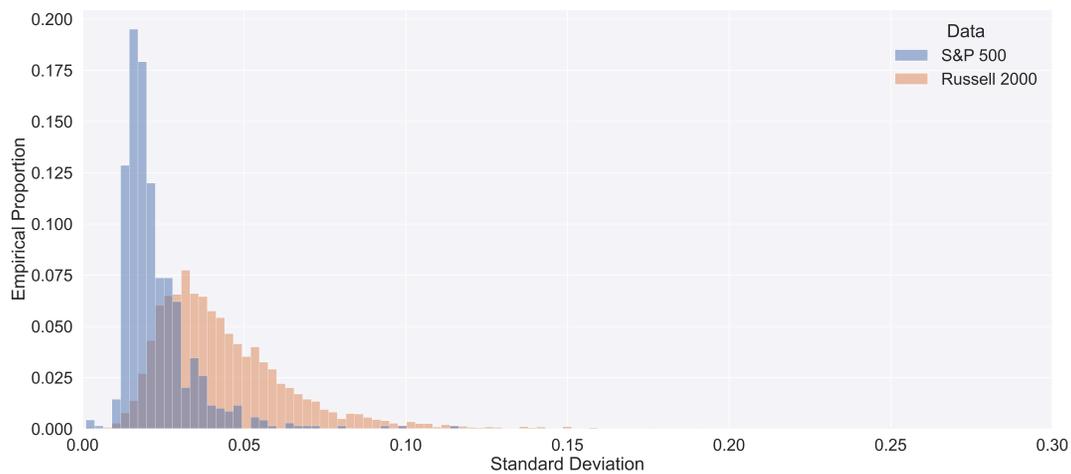} \label{subfig:std}
        }
         \caption{Visualization of the datasets.}
        \label{fig:cor_std}
    \end{figure}
    As discussed earlier, most high-dimensional portfolio selection literature involves empirical studies on the S\&P 500 index, while only a few analyze their strategies on Russell 2000.
     Dealing with the latter requires computational efficiency and numerical stability on high-dimension datasets, which is more challenging.

    The differences between the two datasets are illustrated in Figure \ref{fig:cor_std}.
    It is evident that the prices in the S\&P 500 dataset exhibit higher correlations than prices in the other dataset.
    To maintain a more diversified portfolio, one needs to include more assets in Russell 2000 since the available stocks are less correlated.
    Another difficulty for portfolio management arising in the high-dimensional scenarios is the increasing volatility.
    Although the high volatility may offer more opportunities for day traders by large swings, it hinders steady gains over holding periods.
    In general, stocks of companies with smaller market capitalization will exhibit higher volatility, as shown by the histograms in Figure \ref{subfig:std}.
    The distribution of the volatility of the Russell 2000 dataset has heavier tails than the one of the S\&P 500 dataset.
     The heavy-tailedness also hinders the estimation of the sample covariance matrix, especially when the number of assets can be larger than the number of observations \citep{ke2019user}.
    Overall, the lower correlations and higher volatility make it tricky to identify valuable assets from large portfolios. 
    As a result, data-driven portfolio strategies will include lots of assets to reduce potential risks, which sacrifices the out-of-sample mean-variance efficiency and is also impractical for high-dimensional PM.
    One natural strategy to tackle the previous issues is holding a reasonable number of assets while maintaining diversiﬁed portfolios.
    \begin{figure}[!tp]
        \centering
          \includegraphics[width=0.9\textwidth]{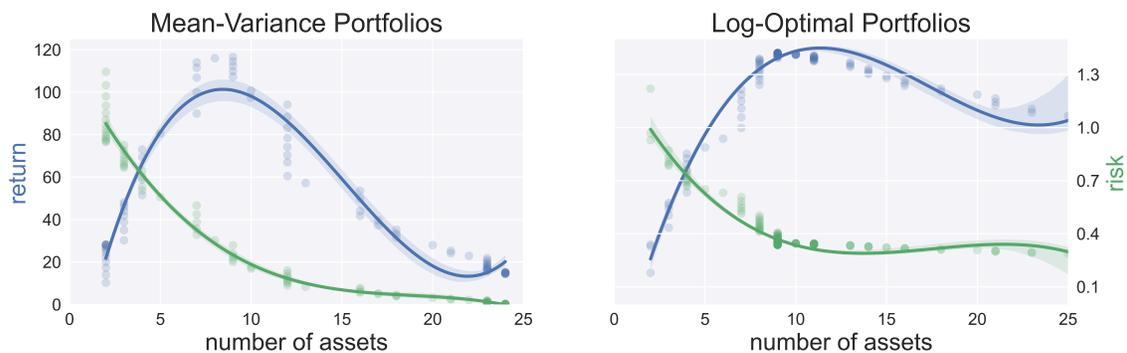}
        \caption{In-sample returns and risks against the number of assets with cubic regression lines plotted.
        The shadow regions denote the estimated 95\% confidence intervals via bootstrapping.
        The portfolios are constructed based on NYSE data with 3680 stocks.}
        \label{fig:frontier}
    \end{figure}

    With such motivation, our portfolio strategy aims to impose cardinality constraints on the number of holding assets while producing satisfying mean-variance efficiency.
    Our approach provides new perspectives to the return-risk trade-offs, where an appropriate selection of $s$ assets can inform us of a better trade-off between return and risk with more stability. 
    As summarized in Figure \ref{fig:frontier}, each point represents one portfolio with a different value of regularization parameters, demonstrating the trade-off between the three factors - risk, return, and $s$ of both the MV approach and our proposed method with logarithm utility. The asset pool consists of 3680 stocks based on NYSE data (See Section \ref{subsec:NYSE}). The left panel of Figure \ref{fig:frontier} is selected from the in-sample efficient frontier of the MV approach. The MV approach and our method show a similar pattern: the better trade-off can be achieved with a relatively small number of assets than the whole asset pool.
    A reasonable amount of diversification contributes to better out-of-sample mean-variance efficiency in this example, indicating that our strategies are applicable for high-dimensional PM.

\section{Statistical Methodology}\label{sec:high-dim}

    \subsection{Sparsity-Induced Portfolios}\label{subsec:eq_l0_l1}

        Suppose we have $d$ assets with price relatives $\cX_1,\ldots,\cX_d$ to be managed,  i.e., random variables $\cX_j$'s are the nonnegative ratios of the price at the end of the day to the price at the beginning of the day for asset $j$.  For the rest of the paper, we assume that $\cX_j\geq \eta_{\min}>0$ and $\EE[\cX_j]<\infty$ for all $j$.
        Let $\bm{\cX}=(\cX_1,\ldots,\cX_d)^{\top}\in\RR^d$ be the price relative vector and $\bw$ be its portfolio allocation vector. The goal of portfolio selection is to maximize the expected utility of wealth:
        \begin{align}
            \min\limits_{\bw\in\cC} \quad -\EE [u(\bw^{\top}\bm{\cX})] \label{opt:expected}
        \end{align}
        where the expectation is taken with respect to $\bm{\cX}$ and $\cC\subseteq\RR^d$ is a feasible set.
        The usual choice of $\cC$ for portfolio selections can be: (1) The subspace $\RR_+$ or $\RR$; (2) The simplex of dimension $d-1$, $\cS_d=\{\bw\in\RR^d: \bw^{\top}\one=1\}$; (3) The probability simplex of dimension $d-1$, $\Delta_d=\{\bw\in\RR^d: \bw\succeq \bm{0},\bw^{\top}\one=1\}$.
        For the first choice of $\cC$, extra norm constraints or regularizers on $\bw$ may be needed to guarantee that the entries of solutions are all finite.
        We will mainly focus on the case when $\cC=\Delta_d$.
        In the expected utility theory, the utility function $u$ quantitatively measures the outcome value to an investor, and some common utility functions are listed in Table \ref{tab:utility_func}.
        We formally define it as follows.
        \begin{table}[!t]\footnotesize
           \begin{tabularx}{0.97\textwidth}{XXl}
                \toprule
                \textbf{Name} & \textbf{Formula} & \textbf{Parameter}\\
                \midrule
                Linear utility & $u(z)=\eta+az$, & $a>0$, $\eta\in\RR$\\
                Quadratic utility & $u(z)=\eta z-\lambda z^2/2$, $z<\eta/\lambda$ & $\lambda>0,\eta\in\RR$\\
                Exponential utility & $u(z)=1-e^{-a (z +\eta)}$, & $a>0$, $\eta\in\RR$\\
                Power utility & $u(z)=(z+\eta)^{1-\lambda}$, & $\lambda\in(0,1),\eta\geq 0$\\
                Logarithmic utility & $u(z)=\log(z+\eta)$, &$\eta\geq0$\\
                \bottomrule
            \end{tabularx}
            \caption{List of common utility functions.}\label{tab:utility_func}
        \end{table}
        \begin{definition}[Utility Functions]\label{def:utility}
            A function $u:\RR_+\rightarrow \RR$ is a utility function if $u$ is increasing and concave on its domain.
        \end{definition}

        For the following section, we use $f$ to denote a general objective function.
        To impose a sparsity structure on the estimated portfolio allocations, one would add a regularization term $\Omega$, resulting in an optimization problem:
        \begin{align}
                \min\limits_{\bw\in\cC}&\quad  f(\bw) + \lambda \Omega(\bw),
            \label{opt:regularized}
        \end{align}
        where $\lambda>0$ is a regularization parameter that controls the sparsity of the allocation and can be picked from the Lasso-type path, and $\Omega$ is a sparsity-induced norm.
        Alternatively, one can consider the constrained problem:
        \begin{align}
                \min\limits_{\bw\in\cC}&\quad  f(\bw) \quad \text{s.t.} \quad \Omega(\bw)\leq s,
            \label{opt:constrained}
        \end{align}
        where $s>0$ is the threshold parameter.
        The two formulations are mathematically equivalent when the objective function and $\Omega$ are convex to $\bw$, and the global minimum of the constrained problem is strictly feasible.
         Let $\Sigma_d^s=\{\bw\in\RR^d:\|\bw\|_0\leq s\}$.
          Considering the cardinality constraints with $\cC=\Delta_d$, we formulate the problem as:
        \begin{align}
            \min\limits_{\bw\in\Delta_d\cap\Sigma_d^s}  f(\bw) ,\label{opt:l0_constrained}
        \end{align}
        which is unfortunately NP-hard even if the expectation can be evaluated efficiently \citep{zhang2019learning}.
         When $f$ is a quadratic function for a mean-variance portfolio, the problem can be cast into a mixed-integer quadratic program (MIQP) by introducing $d$ extra variables.   
       However, it is still computationally intractable and memory-consuming when $d$ is large and the covariance matrix is dense.
        Though many convex relaxation algorithms are proposed for MIQPs, numerical results show that they may provide poor lower bounds and even be incapable of solving medium-sized problems within a reasonable time.
        For example, it may take up to 40 minutes to run only one instance when $d=458$ and $s=7$ \citep{zheng2014improving}.
        One may be interested in whether we can first solve the $l_1$-regularized problem \eqref{opt:l1_regularized}:
        \begin{align}
            \min\limits_{\bw\in\RR_+^d}  f(\bw) + \lambda\|\bw\|_1, \label{opt:l1_regularized}
        \end{align}
        with the non-negativity constraint and then project the estimated allocation onto the $l_1$-norm unit sphere.
        Such an intuitive approach indeed works under mild assumptions.
        To see this, we first introduce the single crossing property, formally defined as below:
        \begin{definition}[Single Crossing]
            Suppose that $Z_1$ and $Z_2$ are two random variables.
            We say that $Z_1$ single crosses $Z_2$ from below if there is a crossing point $c\in\RR$ such that
             \begin{align*}
                 \begin{cases}
                    \overline{F}_{Z_1}(t)\geq \overline{F}_{Z_2}(t) &, \ t\leq c\\
                    \overline{F}_{Z_1}(t)\leq \overline{F}_{Z_2}(t)  &,\ t>c
                 \end{cases},
             \end{align*}
             where $\overline{F}_{Z_1}(t)=\PP(Z_1>t)$ and $\overline{F}_{Z_2}(t)=\PP(Z_2>t)$ are the survival functions of $Z_1$ and $Z_2$.
        \end{definition}
         The single crossing property is widely used in the economics literature; see, for example, \citet{jewitt1987risk,athey2001single,ballester2017single}.
        This sign change property holds when
        (1) one of the two decisions leads to a degenerate outcome, such as in the comparative statics problems, or (2) the decisions are made from some distribution family, such as in the single risky asset portfolio problem.
        An important implication of single crossing is that more risk-averse investors, in the sense of Arrow and Pratt \citep{arrow1971theory}, will hold a larger share of their wealth of the safe asset than the single risky asset \citep{jewitt1987risk}. 
        From the statistical perspective, the single crossing property holds when two random variables come from common distributions, such as normal, lognormal, and truncated normal, under different conditions, as discussed in \citet[Chapter 6]{levy2006stochastic}.
        An intuitive example is given in Figure \ref{fig:single_crossing} when the two univariate random variables $Z_1=\bw_1^{\top}\cX$ and $Z_2=\bw_2^{\top}\cX$ are normally distributed.
        We note that this is only a particular case, and no explicit distribution for $\cX$ is assumed in this paper.
        When the decisions $\bw_1$ and $\bw_2$ are estimated from data, $Z_1$ and $Z_2$ are from two families of distribution functions, with each corresponding pair of realizations crossing no more than once, resulting in the preservation of the risk aversion preference. 
        The following theorem shows that one can solve the $l_0$ constrained problem \eqref{opt:l0_constrained}, by solving the $l_1$ regularized problem \eqref{opt:l1_regularized} under such conditions.

        \begin{figure}[!t]
            \centering
            \includegraphics[width=0.9\textwidth]{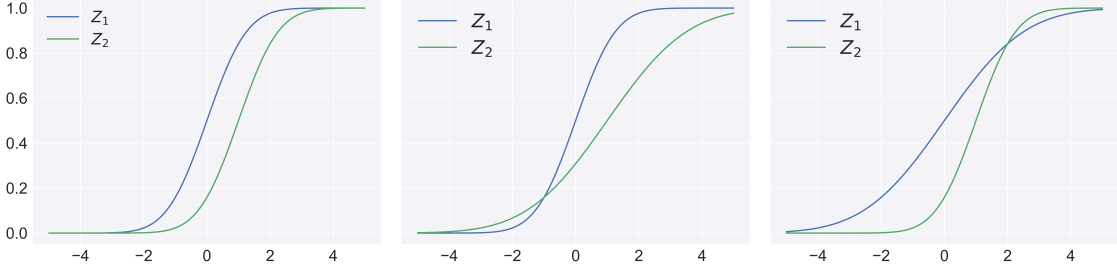}
            \caption{Cumulative distribution functions of $Z_1\sim\mathcal{N}(0,\sigma_1)$ and $Z_2\sim\mathcal{N}(1,\sigma_2)$, when $Z_1$ and $Z_2$ are not crossing ($\sigma_1=\sigma_2$), $Z_1$ single crosses $Z_2$ from above ($\sigma_1<\sigma_2$), and $Z_1$ single crosses $Z_2$ from below ($\sigma_1>\sigma_2$).}
            \label{fig:single_crossing}
        \end{figure}
        
        \begin{theorem}[Equivalence for Expected Utility Maximization]\label{thm:eq_l1_l0_E}
              Suppose that the objective function in both optimization problems \eqref{opt:l1_regularized} and \eqref{opt:l0_constrained} is $f(\bw)=-\EE[u(\bw^{\top}\bm{\cX})]$ with $\cX_j\geq \eta_{\min}>0$ and $\EE[\cX_j]<\infty$ for all $j$.
            Let $\bw_1$ and $\bw_2$ be the solutions to problems \eqref{opt:l1_regularized} and \eqref{opt:l0_constrained} respectively with $s=\|\bw_1\|_0>0$, and define $\tilde{\bw}_1=\bw_1/\|\bw_1\|_1$, $\tilde{\bw}_2=\bw_2\cdot\|\bw_1\|_1$, $Z_i=\bw_i^{\top}\bm{\cX}$ and $\tilde{Z}_i=\tilde{\bw}_i^{\top}\bm{\cX}$ for $i=1,2$.
            Suppose that
            (i) $u$ is strictly increasing and strictly concave, and
            (ii) $\tilde{Z}_1$ and $Z_2$ cross for at most one time ,
            then $\tilde{\bw}_1$ is an solution to problem \eqref{opt:l0_constrained}.
        \end{theorem}
         The single crossing assumption in Theorem \ref{thm:eq_l1_l0_E} states that the better portfolio (in terms of linear utility) can be distinguishable at different levels of random (linear) utility.
        Furthermore, an indifference curve of a better portfolio (in the space of utility levels) crosses that of the other portfolio at most once.
        This assumption captures the idea that a better portfolio is preferred for more random utility. 
        As a result, it is more likely to profit through investing in the better portfolio at higher levels of random utility, thus making it possible to separate the two portfolios.
        In the proof of Theorem \ref{thm:eq_l1_l0_E}, we also show that such preferences can be preserved for other utility functions.
        With such preferences, Theorem \ref{thm:eq_l1_l0_E} guarantees that the two portfolios $\tilde{\bw}_1$ and $\bw_2$ obtain the same expected utility under constraints of \eqref{opt:l0_constrained}.

    \subsection{Sample Average Approximation}\label{subsec:consist}

        Since the optimization over expectation is generally intractable and impractical, a natural way to solve the optimization problem \eqref{opt:expected} is to minimize the empirical sample average objective:

        \begin{align}
            \min_{\bw\in\cC}&\quad h(\bw) \triangleq-\frac{1}{n}\sum_{i=1}^n u(\bw^{\top}\bX_i),\label{opt:SAA}
        \end{align}
        where $\bX_i$'s are independent and identically distributed as $\bm{\cX}$, which is known as the \emph{sample average approximation} (SAA) method in the stochastic optimization literature \citep{shapiro2014lectures}.
        The mean-variance portfolio problems are generally solved implicitly by the SAA method, where the mean and variance of $\bm{\cX}$ are estimated from samples.
        One of the connections between the utility theory and the MV portfolio theory is that the latter's objective is the utility function's \emph{moment approximation} (MA).
        Define $U_{\bw}(\bx)=u(\bw^{\top}\bx)$ and consider the second-order approximation of $-\EE[U_{\bw}(\bm{\cX})]$  :
        \begin{align}
            g_{\bx_0}(\bw, \bmu,\bSigma)\triangleq & -U_{\bw}(\bx_0)- \nabla_{\bx} U_{\bw}(\bx_0)^{\top}(\bmu-\bx_0) -\frac{1}{2}\tr\left( \nabla_{\bx}^2 U_{\bw}(\bx_0)(\bSigma+\bmu \bmu^{\top}) \right) \notag\\
            &\quad +\bx_0^{\top} \nabla_{\bx}^2 U_{\bw}(\bx_0)\bmu - \frac{1}{2} \bx_0^{\top} \nabla_{\bx}^2 U_{\bw}(\bx_0)\bx_0, \label{eq:g_w_mu_sigma}
        \end{align}
        where $\bmu=\EE(\bm{\cX})$ and $\bSigma=\Var(\bm{\cX})$.
        If $\bx_0=\bmu$ and the utility function is quadratic, then equation \eqref{eq:g_w_mu_sigma} turns out be the objective of MV portfolios and the MA methods aim to solve:
        \begin{align}
            \min\limits_{\bw\in\cC}\quad g_{\bx_0}(\bw, \hat{\bmu},\hat{\bSigma}), \label{opt:moment_2}
        \end{align}
        where $\hat{\bmu}$ and $\hat{\bSigma}$ based on samples $\bX_1,\ldots,\bX_n$ of $\bm{\cX}$.
        Solutions to the problem \eqref{opt:moment_2} also refer to the \emph{plug-in estimators} of $\bw$ \citep{cai2020high,ding2021high}.
        While such an approach is justified by classical statistics theory because the plug-in portfolio is an MLE of the optimal portfolio, its out-of-sample performance is generally poor \citep{ao2019approaching}.

        Combining optimization problem \eqref{opt:SAA} with regularization, we next derive the non-asymptotic error bound for coefficients estimation in the ultra-high dimensional settings. 
        To begin with, we consider the constrained variants of population and empirical minimization problems:
        \begin{align}
            \min\limits_{\bw\in\RR_+^d}\quad& \EE[h(\bw)] \quad\text{subject to }\|\bw\|_1\leq R \label{opt:M_estimator_pop}\\
            \min\limits_{\bw\in\RR_+^d}\quad& h(\bw) + \lambda\|\bw\|_1\quad\text{subject to }\|\bw\|_1\leq R.\label{opt:M_estimator_emp}
        \end{align}
        The constraint in the above problems guarantees that the estimator lies in a bounded set, which is helpful for our statistical analysis.
        One can also view problem \eqref{opt:l1_regularized} with objective functions $f(\bw)$ being $\EE[h(\bw)]$ as the Lagrangian formulation of the problem \eqref{opt:M_estimator_pop}.
        For problem \eqref{opt:M_estimator_emp}, the constraint $\|\bw\|_1\leq R$ will have no effect when the regularization parameter $\lambda$ is large enough.
        We denote the solutions to the optimization problems \eqref{opt:M_estimator_pop} and \eqref{opt:M_estimator_emp} by $\bw^*$ and $\hat{\bw}_n$, respectively.
        Before we present our results, we first introduce some matrix notations.
        We use $\bX=
        \left[ \bX_1, \cdots, \bX_n\right]^{\top}\in\RR^{n\times d}$ to denote the sample matrix and $\bX_{\cdot j}=\left[X_{1j} ,\cdots, X_{nj}\right]^{\top}\in\RR^{n}$ to denote its $j$th column for $j=1,\ldots,d$.
        We make the following assumptions.
        
        \begin{assumption}We assume the following holds:
            \begin{enumerate}[(a)]
                \item The utility function $u$ is $L_u$-Lipschitz continuous, differentiable, and concave.
                \item The theoretical minimizer $\bw^*$ is unique and denote its cardinality by $s=\|\bw^*\|_0$.
                
                \item There exists $\nu>0$ such that $(X_{ij}-\nu)$'s are sub-Gaussian with parameter $\sigma$, $\EE[\exp(t(X_{ij}-\nu)]\leq \exp(t^2\sigma^2/2)$, $\forall\ t>0,\ i\in[n],j\in[d].$

                \item (Restricted strong convexity) There exists $\alpha\geq 2$ and $\kappa$ such that almost surely
                $$0<\kappa \leq  \inf\limits_{|\cS|\leq s}\inf\limits_{\bv\in C(\cS,\gamma_1,\gamma_2)}\frac{\bv^{\top}\nabla^2_{\bw}\EE[h(\bw^*)]\bv}{\|\bv\|_2^2}$$
                where $\gamma_1=\frac{\alpha}{\alpha-1}$, $\gamma_2=\frac{\sqrt{s}}{\alpha-1}$, and $C(\cS,\gamma_1,\gamma_2)=\{\bv\in\RR^d\mid \|\bv_{\cS^C}\|_1\leq \gamma_1\|\bv_{\cS}\|_1+\gamma_2\|\bv_{\cS}\|_2\}$.
            \end{enumerate}\label{ass:1}
        \end{assumption}
        From Definition \ref{def:utility}, the concavity in Assumption \ref{ass:1}(a) naturally holds for all utility functions.
        Most of the utility functions in Table \ref{tab:utility_func} such as exponential utility and logarithmic utility with $\eta>0$, are Lipshitz-continuous and differentiable.
        Then it also follows that $H(\bz)=-\sum_{i=1}^nu(z_i)/n$ is $L_z$-Lipshitz continuous with $L_z=L_u/n$.
        Assumption \ref{ass:1}(c) is commonly used in the literature of non-asymptotic high-dimensional statistics \citep{wainwright2019high}.
        Especially, if $\cX$ is a bounded set, then Hoeffding's lemma guarantees that Assumption \ref{ass:1}(c) holds.
        Assumption \ref{ass:1}(d) presumes good curvature for the population cost function within certain region of the parameter space. 
        It is weaker than the one used in \citep[Section 9.8]{wainwright2019high}, since we only need the restricted strong convexity property to hold in the intersection of the unit $l_2$-norm ball and a cone.
        We present our main result in as follows.
        
        \begin{theorem}[Coefficients Estimation Error]\label{thm:est_err}
            Suppose Assumption \ref{ass:1} holds.
            For any $\delta\in(0,1)$ and $\lambda\gtrsim \alpha (\nu\vee 1)L_u\sigma \sqrt{\log(d/\delta)/n}$, with probability at least $1-\delta$, it holds that
            \begin{align*}
                \|\hat{\bw}_n-\bw^*\|_2 & \lesssim \cO(\kappa^{-1}(\nu\vee 1)L_u\sigma\sqrt{s\log(d/\delta)/n}).
            \end{align*}
        \end{theorem}
        Theorem \ref{thm:est_err} says that if the sparsity level and the dimension satisfy that $s\log d=o(n)$, the coefficients estimation error vanishes when both $n$ and $d$ go to infinity.
        Our error rate $\cO(\sqrt{s\log d /n})$ also matches the rate of Lasso estimator for linear regression.
        As an immediate consequence of Theorem \ref{thm:est_err} by Cauchy-Schwartz inequality, the portfolio performances of the SAA solutions also have finite-sample error bounds as stated in Corollary \ref{cor:SAA_perfomance}.
        \begin{corollary}[SAA Performance]\label{cor:SAA_perfomance}Under the same conditions as in Theorem \ref{thm:est_err}, with probability at least $1-\delta$, it holds that\begin{align*}
                |\hat{\bw}_n^{\top}\bmu-\bw^{*\top}\bmu|&\lesssim \cO(\|\bmu\|_2\kappa^{-1}(\nu\vee 1)L_u\sigma\sqrt{s\log(d/\delta)/n})\\
                |\hat{\bw}_n^{\top}\bSigma\hat{\bw}_n-\bw^{*\top}\bSigma\bw^{*}|&\lesssim \cO(R\|\bSigma\|_2\kappa^{-1}(\nu\vee 1)L_u\sigma\sqrt{s\log(d/\delta)/n}).
            \end{align*}
        \end{corollary}

        Combining the equivalence of $l_0$ constraints and $l_1$ regularizers with the asymptotic optimality of the latter, solving the $l_1$-regularized SAA problem is provably equivalent to solving the $l_0$-constrained EUM problem. It is beneficial in two ways.
        Firstly, we no longer deal with expectation as in the EUM problem, which requires the knowledge of the distribution of $\bm{\cX}$.
        Secondly, we work on a much easier convex optimization problem instead of a nonconvex problem.
        Thus, projected gradient descent (Section \ref{subsec:algo}) can be efficiently employed and safe screening (Section \ref{subsec:screen}) can be used to filter unimportant assets and speed up the process.

    \subsection{Feature Screening}\label{subsec:screen}
        Turning the nonconvex optimization problem \eqref{opt:l0_constrained} into the convex optimization problem \eqref{opt:l1_regularized} is beneficial for computation.
        However, it is still computationally challenging when the dimension $d$ is large.
        Fortunately, the screening techniques developed for the convex optimization problem can be used to alleviate this problem.
        Building on the idea of \citet{dantas2021expanding}, we propose safe screening methods for the problem \eqref{opt:l1_regularized} and further derive efficient primal updates.
        We only need the following common assumption about the objective function $h(\bw)$ for the remaining context to establish our theoretical results.
        \begin{assumption}[Convexity and Smoothness]\label{ass:conv_smooth}
            The objective function $h:\RR^d\rightarrow\RR$ is proper convex and has Lipschitz continuous gradients.
        \end{assumption}
        Similar to Lasso problems, the maximum regularization parameter $\lambda_{\max}\triangleq\|\nabla h(\zero_d) \|_{\infty}$ makes the solution to the optimization problem \eqref{opt:l1_regularized} all zeros.
        Hence, we would focus on the case when $\lambda<\lambda_{\max}$.
        The convexity property naturally yields the primal and dual formulations of the convex optimization problem.
        \begin{theorem}[Fenchel-Rockafellar Duality]\label{thm:primal_dual}
            The primal and dual formulations of the optimization problem \eqref{opt:l1_regularized} is given by
            \begin{align*}
                \bw^*&\in \argmin_{\bw\in \cC_P} \cP_{\lambda}(\bw) \triangleq  - \frac{1}{n}\sum\limits_{j=1}^n u(\bw^{\top}\bX_j)+\lambda \|\bw\|_1,\quad
                \btheta^* \in \argmax\limits_{\btheta\in \cC_D}\cD_{\lambda}(\btheta) \triangleq   \frac{1}{n}\sum\limits_{j=1}^n u^*( n\lambda\theta_j),
            \end{align*}
            where $u^*$ is the convex conjugate of $u$, $\cC_P=\dom(\cP_{\lambda})\cap\RR^d_+$ and $\cC_D=\dom(\cD_{\lambda})\cap \{\btheta\in\RR^n\mid  \|\phi(\bX^{\top}\btheta)\|_{\infty}\leq 1\}$ are the primal and dual feasible sets and the function $\phi(x)=\max\{x,0\}$ and inequality relation ``$\leq$'' are applied element-wisely.
            For any primal-dual feasible point $(\bw^*,\btheta^*)\in \cC_P \times \cC_D$, the optimality condition reads that,
            \begin{flalign}
                \text{(Primal-Dual Link)} &&  &-\lambda \btheta^*=\nabla_{\bz}H(\bX\bw^*) \label{eq:primal_dual_link}\\
                \text{(Subdifferential Inclusion)} && &\begin{cases}
                \|\phi(\bX_{\cdot j}^{\top}\btheta^*)\|_{\infty}\leq 1,&w_j^*=0\\
                \|\phi(\bX_{\cdot j}^{\top}\btheta^*)\|_{\infty}=1,\ \btheta^{*\top}\bX_{\cdot j} w_j^*=|w_j^*|,&w_j^*\neq0.
                \end{cases} \label{eq:subdiff}
            \end{flalign}
        \end{theorem}
        By Theorem \ref{thm:primal_dual}, we can compute the duality gap
        $\Gap_{\lambda}(\bw,\btheta)\triangleq\cP_{\lambda}(\bw)-\cD_{\lambda}(\btheta)$, which is nonnegative due to the weak duality property and quantifies the suboptimality of the feasible primal-dual pair for the optimization problem \eqref{opt:l1_regularized}.
        Moreover, the duality gap can be used to construct a safe region that contains the optimal dual variables, as shown in Theorem \ref{thm:safe_region}.
        By \citep[Theorem 6]{ndiaye2017gap}, $\cD_{\lambda}$ is strongly concave if $u_i$'s are differentiable with Lipschitz gradients.
        Thus, under Assumption \ref{ass:conv_smooth}, the constructed sphere is guaranteed to be safe and called \emph{gap safe sphere}.
        To obtain feasible dual variables, we can apply dual scaling
        $\bXi(\bz)= \bz/\max\{\|\phi(\bX^{\top}\bz)\|_{\infty},1\}$ to $\btheta$ so that a meaningful safe region can be computed.

        \begin{theorem}[Safe Region]\label{thm:safe_region}
            For any primal-dual feasible point $(\bw,\btheta)\in \cC_{P}\times\cC_{D}$, we have that $\btheta^*\in \cB(\btheta, r)$ where $r=\sqrt{2\Gap_{\lambda}(\bw,\btheta)/\alpha}$ if $\cD_{\lambda}$ is $\alpha$-strongly concave.
            Furthermore, the safe screening rule for each feature $j\in[d]$ is given by ``$\phi(\bX_{\cdot j}^{\top}\btheta)+r\|\bX_{\cdot j}\|_{2}<1
            \ \Longrightarrow\ w^*_j=0.$''
        \end{theorem}

        \subsection{Algorithm}\label{subsec:algo}
        
        \begin{algorithm}[t]
          \caption{Portfolio Selection on $\Delta_d$ with Feature Screening.}\label{algo:spo}
          \begin{algorithmic}[1]
            \REQUIRE The utility function $u$, the relative price matrix $\bX\in\RR^{n\times d}$, the regularization parameter $\lambda$, the step size $\iota$, and the initial value of $\bw^{(0)}$.
            \STATE Initialize the step number $t=1$, the screening set $\cS^{(0)}=\varnothing$ and the active set $\cA^{(0)}=[d]$.
            Form the primal and dual objective functions $$\cP_{\lambda}(\bw)=h(\bw)+\lambda\|\bw\|_1,\qquad\cD_{\lambda}(\btheta)=h^*(\btheta),\qquad \forall\ (\bw,\btheta)\in\cC_P\times \cC_D,$$ where $h(\bw)=H(\bX\bw)=\sum_{j=1}^nu(\bw^{\top}\bX_j)/n$.
            Set $\alpha=\max_{j}[-\nabla^2\cD_{\lambda}(\zero_n)]_{jj}$. \label{algo:init}
            \WHILE{not terminate}
                \STATE Allocate memory for $\bw^{(t)}$ (e.g. modified in place of $\bw^{(t-1)}$).

                \noindent \gray{$\diamond$ Primal Update}
                \STATE $\bw^{(t)}_{\cA^{(t)}}= \Prox_{\iota\lambda\Omega}(\bw_{\cA^{(t)}}^{(t-1)} - \iota \nabla_{\bw_{\cA^{(t)}}}h(\bw^{(t-1)}))$ where $\Omega(\bv)=\|\bv\|_1+\mathds{1}_{\RR^{|\cA^{(t)}|}_+}(\bv)$. \label{algo:primal_update}

                \noindent \gray{$\diamond$ Dual Evaluate And Scaling}
                \STATE $\btheta^{(t)}= \bXi( - \nabla_{\bz}H(\bX\bw^{(t)})/\lambda)$. \label{algo:dual_update}

                \noindent \gray{$\diamond$ Screening}
                \STATE $\Gap_{\lambda}^{(t)}= \cP_{\lambda}(\bw^{(t)})-\cD_{\lambda}(\btheta^{(t)})$.
                
                \STATE $r^{(t)}= \sqrt{2\Gap_{\lambda}^{(t)}/\alpha}$.
                \STATE $\cS^{(t)}=\cS^{(t-1)}\cup\{j\in\cA^{(t-1)}:\phi(\bX_{\cdot j}^{\top}\btheta^{(t)})+r^{(t)}\|\bX_{\cdot j}\|_{2}< 1\}$.
                \STATE $\cA^{(t)}=[d]\setminus\cS^{(t)}$.
                \STATE $\bw^{(t)}_{\cS^{(t)}}= \zero_{|\cS^{(t)}|}$.
                \STATE $t=t+1$.
            \ENDWHILE

            \STATE $\hat{\bw}=\bw^{(t)}/\|\bw^{(t)}\|_1$

            \ENSURE The portfolio allocation $\hat{\bw}$.
          \end{algorithmic}
        \end{algorithm}
        
        Our main algorithm is described in Algorithm \ref{algo:spo}.
        For any utility function $u$ satisfying Assumption \ref{ass:conv_smooth} and data matrix $\bX$, we first define the primal and dual objectives as on line \ref{algo:init} of Algorithm \ref{algo:spo}.
        We perform the primal and dual updates in each iteration, followed by the gap safe screening.
        Recall that the proximal operator is defined as $\Prox_g(x)\triangleq \argmin_y g(y)+ \|x-y\|_2^2/2$ for any function $g$.
        On line \ref{algo:primal_update}, we perform one-step gradient descent on $\bw^{(t)}$ followed by applying the proximal operator of $\Omega(\bv)=\|\bv\|_1+\mathds{1}_{\RR^{|\cA^{(t)}|}_+}(\bv)$, where $\cA^{(t)}$ is the active set of stocks in iteration $t$ and $\mathds{1}_{\cC}(\bv)$ is the indicator function on a set $\cC$ that takes value 0 when $\bv\in\cC$ and infinity otherwise.
        Noted that the proximal update is not the same as soft-thresholding for Lasso because of the non-negativity constraint.
        Although there is no closed-form for the proximal operator, we can still evaluate it efficiently, as discussed in Supplement Materials.
        On line \ref{algo:dual_update} of Algorithm \ref{algo:spo}, we compute the dual variable defined in \eqref{eq:primal_dual_link} and apply the scaling function $\bXi(\bz)= \bz/\max\{\|\phi(\bX^{\top}\bz)\|_{\infty},1\}$ so that the dual variable is feasible.
        Finally, we perform gap-safe screening and update the set of active coordinates.
        For the next iteration, we just need to repeat the above process for the subvector $\bw_{\cA^{(t)}}$ since any coordinate of $\bw$ in the complement of $\cA^{(t)}$ is guaranteed to be zero.
        Though Algorithm \ref{algo:spo} is general and inclusive, we will specialize it to logarithm utility and exponential utility in Section \ref{sec:utility_case}, and some valuable ingredients are summarized in Supplement Materials.

        Combining the proximal gradient descent algorithm with the screening rules, the convergence guarantee as in Theorem \ref{thm:convergence} can also be established for Algorithm \ref{algo:spo}.
        This implies that Algorithm \ref{algo:spo} has a convergence rate $\cO(1/t)$ or $\cO(1/\epsilon)$ where $t$ is the number of iterations and $\epsilon$ is the tolerance.
        We also noted that it is not necessary to perform screening at each iteration in Algorithm \ref{algo:spo}.
        So we further add an option in our implementation to run screening for every $i_s$ iteration, where $i_s$ is a parameter.
        In general, $i_s$ should not scale linearly with the maximum number of iterations so that the time consumption on screening will not be significant compared to proximal gradient descent steps.
        
        \begin{theorem}[Convergence]\label{thm:convergence}
            Suppose that Assumption \ref{ass:conv_smooth} holds and $\bw^*$ is defined in Theorem \ref{thm:primal_dual}.
            If the step size $\iota\leq 1/L_{\nabla h}$, then the estimate $\bw^{(t)}$ at the $t$-th step of Algorithm \ref{algo:spo} satisfies
            \begin{align*}
                |\cP_{\lambda}(\bw^{(t)})-\cP_{\lambda}(\bw^*)|\leq \frac{\|\bw^{(0)}-\bw^*\|_2^2}{2\iota t}.
            \end{align*}
        \end{theorem}

\section{Empirical Studies}\label{sec:empirical_study}
    \subsection{Notable Particular Cases}\label{sec:utility_case}
       Motivated by the Kelly strategy \citep{kelly}, the growth rate of stocks is defined as $\mathbb{E}[\log(\bw^{\top}\bm{\cX})]$. The strategy maximizing the growth rate is the so-called growth-optimal or log-optimal strategy. It is equivalent to maximizing the expected logarithmic utility when we set $\eta=0$. The log-optimal strategy enjoys the benefits in the long run \citep{algoet1988asymptotic}, i.e., it accumulates more return than any other portfolio with probability 1. It sounds like an encouraging result, but it is unrealistic since the investors may not have such a ``long'' lifetime \citep{rubinstein1991continuously}.
        Although the shortcomings of the log-optimal strategy cannot be ignored, the portfolio related to logarithmic utility still plays a vital role in portfolio management \citep{gyorfi2012empirical,rujeerapaiboon2016robust}.
        For numerical stability purposes and practical concerns, in practice, researchers also consider variants of the log-optimal strategy with $\eta>0$ and hyperbolic absolute risk aversion (HARA) property $-u''(z)/u'(z)=(z+\eta)^{-1}$ \citep{kroll1984mean,ccanakouglu2010portfolio}.
        In Section \ref{sec:empirical_study}, we evaluate the portfolio performance by maximizing the logarithmic expected utility based on our methods with $\eta=\eta_{\min}$ estimated from real data.

        Another class of portfolio strategy is the exponential utility portfolio, which is commonly used as its optimization will lead to celebrated mean-variance optimization problems when the returns are normally distributed. Moreover, the exponential utility function implies constant absolute risk aversion (CARA), where $a$ is a constant over time, while the logarithmic utility function decreases absolute risk aversion.
        Under normality of the return vector, the larger the risk aversion parameter $a$, the smaller holding of each asset is.
        Combining with $l_1$ regularization as in optimization problem \eqref{opt:l1_regularized}, this implies that the number of assets of the optimal portfolio will be smaller for larger $a$. For our experiments in Section \ref{sec:empirical_study}, we evaluate the portfolio performance with several risk aversion parameters $a\in\{0.05,\ 0.1,\ 0.05,\ 1,\ 1.5\}$.
        For the rest of the paper, we use LOG and EXP-$a$ to denote portfolio strategies with logarithmic and exponential utility with aversion parameter $a$, respectively.

    \subsection{Algorithm Performance on NYSE Stocks}\label{subsec:NYSE}
        In this section, we intend to examine the effect of safe gap screening on the performance of Algorithm \ref{algo:spo} regarding screening ratios, time complexity, and convergence rate in the high-dimensional scenarios on real data, while no portfolios are constructed here.
        We collect monthly returns of NYSE stocks from the Center for Research in Security Prices (CRSP) database, keep as many stocks as possible by including 3680 stocks with missing records less than of 30\% from 2016 to 2018, and then fill the missing returns with zeros.
        Finally, we get a sample matrix $\bX$ with a shape of $24\times 3680$.
        \begin{figure}[t]
            \centering
            \subfigure[{Screening ratio against iterations for $\log_{10}(\lambda/\lambda_{\max})\in[-3,0]$ for LOG and EXP-1.00 utility functions (The lighter, the more screened features).}]{
              \includegraphics[width=0.9\textwidth]{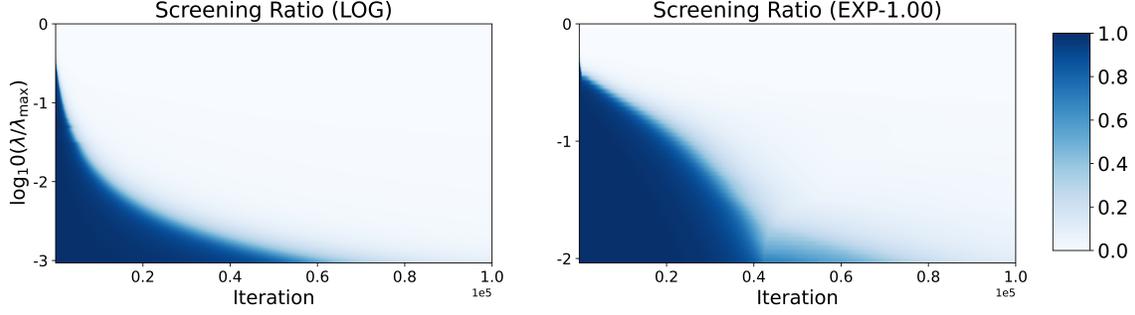} \label{fig:simulation_a}
            }

            \subfigure[The relative execution time (the ratio of screened time to unscreened time) against $\log_{10}(\lambda/\lambda_{\max})$, and the convergence rate (duality gap) against execution time when $\lambda/\lambda_{\max}=5e-1$, for LOG and EXP-1.00 utility functions.]{
              \includegraphics[width=0.9\textwidth]{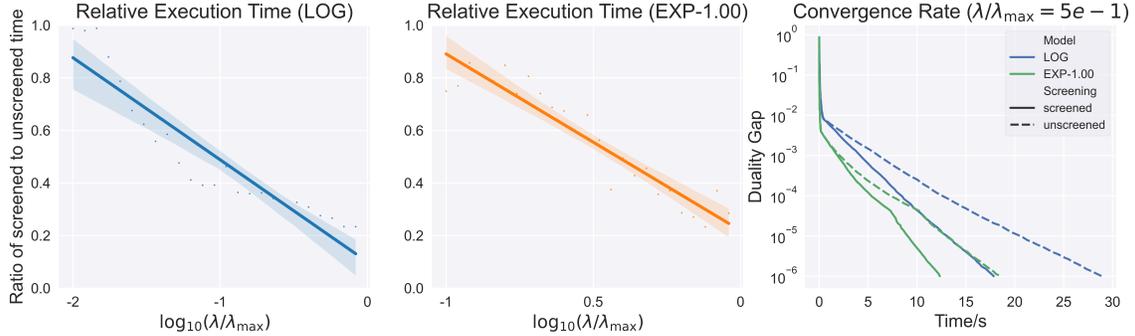} \label{fig:simulation_b}
            }

            \caption{Simulation results on NYSE historical data when $n=24$ and $d=3680$.}
            \label{fig:simulation}
        \end{figure}

               Figure \ref{fig:simulation_a} shows the screening ratio (defined as the number of screened features divided by the total number of features) against the number of iterations.
        The colors indicate the value of the screening ratio.
        The lighter, the larger ratio of screened features is.
        We observe that most regions are nearly white, meaning that the portfolio allocation vector returned by Algorithm \ref{algo:spo} is highly sparse.
        When the level of regularization $\log_{10}(\lambda/\lambda_{\max})$ gets more prominent, so does the effect of the $l_1$ regularization.
        The result indicates that with strong regularization, Algorithm \ref{algo:spo} can distinguish the most useless features at the early stage of optimization, which benefits the optimization by focusing only on the unscreened features, as we shall see later.
        
        The first two plots of Figure \ref{fig:simulation_b} show the relative execution time (defined as the execution time of a model with screening divided by the time without screening) against $\log_{10}(\lambda/\lambda_{\max})$.
        It is clear that with more robust regularization, the optimal portfolio allocation has more zero entries, in which case our algorithm benefits more from the screening procedure.
        Thus, we see that the relative execution time decreases as $\lambda$ increases. The last plot of Figure \ref{fig:simulation_b} presents the convergence rate (duality gap) against execution time for the models with and without screening when $\lambda/\lambda_{\max}=10^{-5}$.
        We see no discernible difference in the long-run duality gap, but the unscreened algorithm takes 30 and 19 seconds, whereas the screened one takes 18 and 12 seconds to reach the convergence threshold of $10^{-6}$ for LOG and EXP-1.00, respectively, exhibiting a large amount of timesaving (about $40\%$).
        
        We note that Algorithm \ref{algo:spo} considers the specific structure of the optimization problem. 
        For example, the Hessian in \Cref{tab:ingredients} is diagonal, meaning that the second-order information can be efficiently exploited. 
        Hence, it is perceptible that our algorithm is more efficient than generic convex solvers.
        The experiment results in Supplement Materials (Section E) also confirm the efficiency of our algorithm compared to generic convex solvers.

    \subsection{Portfolio Selection on S\&P 500 Index Constituents}
        \label{subsec:SP500}

        We first investigate the performance of different strategies on the S\&P 500 index constituents.

        \subsubsection{Comparing with Benchmark Portfolios} \label{subsec:benchmark}
            We compare our methods with the following three types of portfolio strategies:
            (1) Equally-weighted (EW) portfolio.
            (2) Global minimum variance (GMV) portfolio using sample covariance matrix (GMV-P), the linear shrinkage estimation of the covariance matrix $\bSigma$ (GMV-LS), and non-linear shrinkage estimation of $\bSigma$ (GMV-NLS).
            (3) Mean-variance (MV) portfolios with the covariance estimators analogously (MV-P, MV-LS, and MV-NLS).
            \begin{align}
                \min\limits_{\bw\in \Delta_d}&\quad \bw^{\top}\hat{\bSigma}\bw,\tag{GMV Portfolio}\label{opt:GMV}\\
                \min\limits_{\bw\in \Delta_d}&\quad -\bw^{\top}\hat{\bmu} + \lambda_{\text{MV}}\bw^{\top}\hat{\bSigma}\bw.  \tag{MV Portfolio}\label{opt:MV}
            \end{align}

            We perform the out-of-sample portfolio evaluation, beginning in 2011 and rebalancing per quarter (63 trading days, denote d by $n_{\text{hold}}$).
            The new model parameters are estimated from 120 training samples before the rebalancing trade day. The resulting sample matrices are of shape $120\times d$, where $d$ varies in different periods (median:436, maximum:454), which is naturally challenging for estimation and optimization because of its high dimensionality (the number of features is much larger than the number of samples).
            We use 5-fold time-series cross-validation to select the hyperparameters, analogous to $K$-fold cross-validation without shuffling but only with training samples prior to the evaluation set in each fold.   
            The regularization parameter $\lambda$ from
            the set of 100 points spaced evenly on a log scale from $\log_{10}(\lambda_{\max})-2$ to $\log_{10}(\lambda_{\max})$ for our models, and
            $\lambda_{\text{MV}}$ from
            the set of 100 points uniformly distributed in spaced evenly on a log scale from $-3$ to $2$ for MV portfolios.

           For our models, we set the maximum iterations to be $10^{4}$ and the convergence tolerance to be $10^{-5}$ during cross-validation.
            After the parameter $\lambda$ is selected from cross-validation, we set them to $10^{5}$ and $10^{-8}$, respectively.
            We perform screening every $30$ iterations for all the experiments in this section.
            For GMV and MV portfolios, we use solver \textsc{Cvxopt} \citep{andersen2015cvxopt} to solve the corresponding convex optimization problems \eqref{opt:GMV} and \eqref{opt:MV} by quadratic programming.
            The parameters for the solver are specified as: \texttt{abstol}=$10^{-12}$, \texttt{reltol}=$10^{-11}$, \texttt{maxiters}=$10^4$, \texttt{feastol}=$10^{-16}$.

            \begin{table}[!t]\footnotesize 
                \begin{tabularx}{\textwidth}{lPPPPP}
                    \toprule
                    \textbf{Method} & \textbf{Return} &  \textbf{Maximum Drawdown} & \textbf{Sharpe Ratio} & \textbf{Sortino Ratio} & \textbf{Avg. Num. of Assets}\\
                    \midrule
                    Benchmark &&&\\
                    \hline
                    \hline\addlinespace[0.2em]
                    EW & 2.3231 & 0.3941 & 0.7084 & 0.9848 & 437 \\
                    GMV-P & 2.1061 & 0.2483 & 0.9487 & 1.4924 & 281 \\
                    GMV-LS & 2.4796 & 0.2929 & 0.9800 & 1.4648 & 258 \\
                    GMV-NLS & 2.3775 & 0.3223 & 0.9410 & 1.3577 & 323 \\
                    MV-P & 1.5481 & 0.2343 & 0.5287 & 0.7654 & 211 \\
                    MV-LS & 1.9540 & 0.2363 & 0.5981 & 0.8735 & 170 \\
                    MV-NLS & 4.0390 & 0.2341 & 0.8681 & 1.2980 & 118 \\
                    \midrule
                    Our methods &&&\\
                    \hline
                    \hline\addlinespace[0.2em]
                    LOG-1.00 & 8.9922 & 0.3194 & 0.9953 & 1.5743 & 20 \\
                    EXP-0.05 & 6.2313 & 0.3504 & 0.8310 & 1.4563 & 37 \\
                    EXP-0.10 & 5.7691 & 0.3027 & 0.8729 & 1.3894 & 23 \\
                    EXP-0.50 & 8.5086 & 0.3558 & 0.9509 & 1.5421 & 7 \\
                    EXP-1.00 & 8.7293 & 0.3558 & 0.9777 & 1.5679 & 6 \\
                    EXP-1.50 & 8.1633 & 0.3558 & 0.9637 & 1.5077 & 5 \\
                    \bottomrule
                \end{tabularx}
                \caption{Out-of-sample results (without transaction fees) on S\&P 500 from 2011 to 2020.}\label{tab:sp500}
            \end{table}

          We measure the out-of-sample performance of each model by three indicators on the daily returns: (1) the accumulated return (RET)\footnote{We obtain the risk-free rate from the Kenneth French Data Library to compute excess returns.}, measuring how much wealth the portfolio strategy earns; (2) the maximum drawdown (MDD), the worst possible risk the investors may face when adopting the strategy; (3) the annualized Sharpe ratio (SR), describing how much excess returns relative to its volatility the portfolio offers and computed from (out-of-sample) daily returns; and (4) the annualized Sortino ratio (SoR), which replaces the overall volatility in SR with the downside deviation \citep{foster2003performance}.

            The results of the S\&P 500 datasets are summarized in Table \ref{tab:sp500}.
            The baseline portfolio strategy EW invests every available asset with the same amount.
            We see that EW has a return of 2.3231, which is 0.99 more than the excess return of the S\&P index (whose open and close prices in the evaluation period are 1257.62 and 3756.07, respectively) because we have filtered stocks not listed in the index or having missing records in the preceding year.
            However, the EW strategy has the most significant drawdown among all methods, indicating the validity of both MV-based strategies and our methods.
            Among the benchmark methods, we observe that the GMV-NLS and MV-NLS strategies perform the best in SR and SoR.
            Our methods outperform them slightly in SR while significantly in SoR.
            Especially, the LOG strategy outperforms the MV-P strategy with an increase of $0.56$ (and $0.98$) in annualized SR (and SoR), indicating $56\%$ (and 98\%) excess return adjusted for risks per year.
            For the EXP-$a$ portfolios, the average number of holding assets decreases as $a$ increases, but the out-of-sample return-risk efficiency is not affected too much.

            \begin{figure}[t]
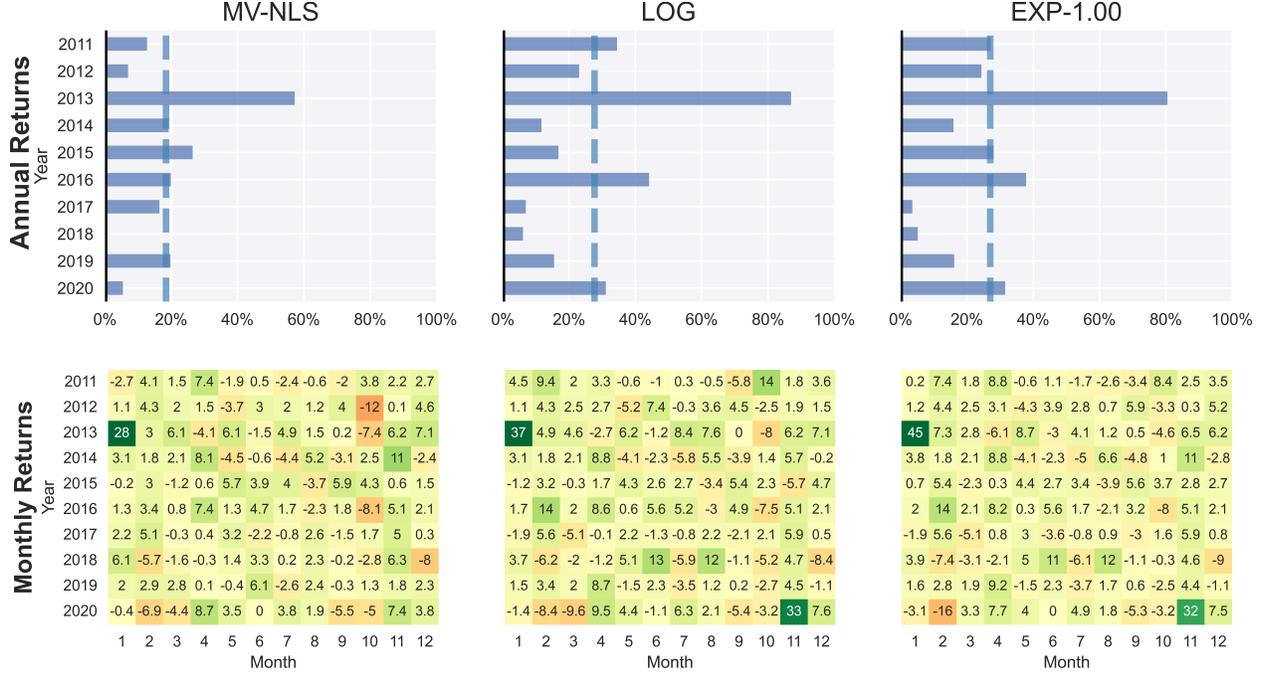

                \centering
                  \includegraphics[width=\textwidth]{img/sp500_returns_1.pdf}
                  \vspace{0.5mm}

                  \includegraphics[width=0.98\textwidth]{img/sp500_returns_2.pdf}\hspace{2mm}
                \caption{Annual and monthly returns on S\&P 500 from 2011 to 2020 with $n=120$, $d \le 454$, $n$\textsubscript{hold}$=63$ for MV-NLS, LOG and EXP-1.00. The dash lines denote the mean annual returns.}
                \label{fig:sp500_retruns}
            \end{figure}

            To evaluate the portfolio strategies in detail, we further visualize the annual and monthly returns for the best benchmark method (MV-NLS) with the most negligible drawdown and our methods (LOG and EXP-1.00) in Figure \ref{fig:sp500_retruns}.
            The coarse-grained plots of annual returns show how each strategy performs over the years.
            The dashed line indicates the average annual return.
            We noted that the huge profit in 2013 for all strategies came from holding the Netflix shares from January 4, 2013, to March 7, 2013.
            More specifically, the Netflix shares surged on surprise profit for about 70\% on January 23, 2013.
            Nevertheless, our methods still perform better than MV-NLS even if we exclude this abnormal period.

            The fine-grained plots of monthly returns show that the impacts of the market on the effectiveness of all strategies are similar.
            For example, all three methods failed to profit from June to July 2017 and from January to February 2020, when the US stock market faced downside risks.
            From the plots, we can also see that the LOG strategy is less risk-averse because it gains enormous profits for some months and loses more for others.
            In the long run, however, it earns 50\% more than MV-NLS in ten years.

            To consider the transaction fees, we define the net returns as:
        \begin{align*}
            r_{\text{net}}^{(t)} = \left(1-\sum\limits_{j=1}^d c|w_{j}^{(t)}-w_{j}^{(t-1)}|- \sum\limits_{j=1}^d c'\mathds{1}_{\{w_{j}^{(t)}\neq w_{j}^{(t-1)}\}}\right)(1+r^{(t)}) -1,
        \end{align*}
        where $\bw^{(t)}$ and $r^{(t)}$ are the allocation vector and the return during the $t$th trading period, and the constant rate $c$ is chosen to be $0.1\%$ as suggested by \citet{robert2012measuring} and $c'$ is chosen to be $0.001\%$ to account for trading activity fees.
        As shown in Table \ref{tab:sp500_tranfee}, the LOG strategy and most EXP strategies outperform benchmark methods when considering transaction fees.

        \begin{table}[!t]\footnotesize 
            \begin{tabularx}{\textwidth}{lPPPPP}
                \toprule
                \textbf{Method} & \textbf{Return} &  \textbf{Maximum Drawdown} & \textbf{Sharpe Ratio} & \textbf{Sortino Ratio} & \textbf{Avg. Num. of Assets}\\
                \midrule
                Benchmark &&&\\
                \hline
                \hline\addlinespace[0.2em]
                EW & 1.8648 & 0.3941 & 0.6283 & 0.8708 & 437 \\
                GMV-P & 1.5428 & 0.2483 & 0.7805 & 1.2172 & 281 \\
                GMV-LS & 1.8424 & 0.2929 & 0.8220 & 1.2208 & 258 \\
                GMV-NLS & 1.7499 & 0.3223 & 0.7838 & 1.1240 & 323 \\
                MV-P & 1.0943 & 0.2415 & 0.4336 & 0.6252 & 211 \\
                MV-LS & 1.4417 & 0.2550 & 0.5062 & 0.7363 & 170 \\
                MV-NLS & 3.3271 & 0.2341 & 0.7924 & 1.1809 & 118 \\\midrule
                Our methods &&&\\
                \hline
                \hline\addlinespace[0.2em]
                LOG-1.00 & 8.1604 & 0.3240 & 0.9613 & 1.5183 & 20 \\
                EXP-0.05 & 5.5768 & 0.3626 & 0.7961 & 1.3927 & 37 \\
                EXP-0.10 & 5.2082 & 0.3073 & 0.8375 & 1.3311 & 23 \\
                EXP-0.50 & 7.7953 & 0.3558 & 0.9215 & 1.4922 & 7 \\
                EXP-1.00 & 8.0122 & 0.3558 & 0.9481 & 1.5184 & 6 \\
                EXP-1.50 & 7.4773 & 0.3558 & 0.9333 & 1.4584 & 5 \\
                \bottomrule
            \end{tabularx}
            \caption{Out-of-sample results (with transaction fees) on S\&P 500 from 2011 to 2020.}
            \label{tab:sp500_tranfee}
        \end{table}

        \subsubsection{Portfolio with Cardinality Constraints}\label{subsubsec:sp500_card}

            For real-life asset allocations, investors have to face more significant systematic risks to have potentially higher returns, but they can reduce the exposure to unsystematic risks by diversifying their investments.
            While there is no single answer to what is a reasonable number, one can always explore possible numbers through historical data.
            We evaluate our methods on the same data set in Section \ref{subsec:benchmark} but with restrictions on the asset allocations.
            Formally, we aim to solve the problem \eqref{opt:l0_constrained} given different numbers of assets $s$.
            Due to the equivalence relationship, we can portray the Lasso-type path efficiently for the problem \eqref{opt:l1_regularized}.

            Figure \ref{fig:sp500_n} shows the portfolio results constrained on the maximum number of assets.
            We can see that both logarithm and exponential utility functions have similar performance when restricting the maximum number of assets to be invested.
            For most strategies, the annualized SRs are more than 0.9 when $s$ is larger than 20 and become relatively stable as $s$ further increases.
            From the plots of the average number of assets, we can also see how the risk-aversion parameter of the exponential utility influences the portfolio allocations.
            As the risk-aversion parameter increases, the portfolio strategy tends to put fewer investments on risky assets, and thus fewer assets will be included when the regularization weakens.

            \begin{figure}[t]
                \centering
                  \includegraphics[width=\textwidth]{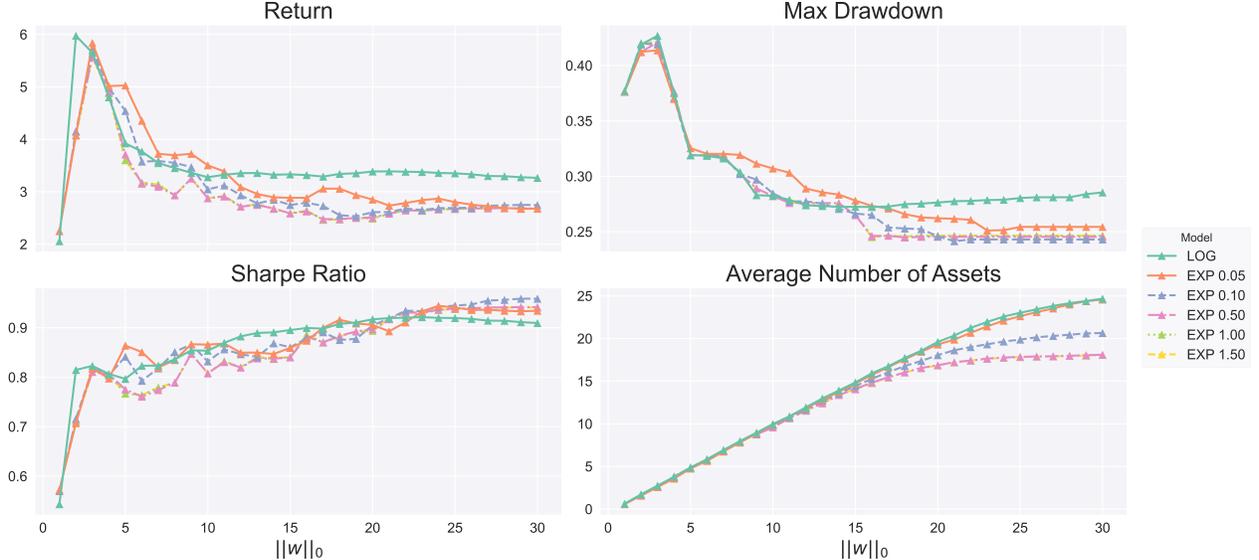}
                \caption{Portfolio performances on S\&P 500 historical data from 2011 to 2020 constrained on maximum number of assets, when $n=120$, $d\leq 454$ and $n$\textsubscript{hold}$=63$ days.}
                \label{fig:sp500_n}
            \end{figure}

    \subsection{Portfolio Selection on Russell 2000 Index Constituents}\label{subsec:russell}
        When faced with larger stock pools and longer time horizons, we analyze our portfolio techniques in this subsection to see how they fare.
        We execute the out-of-sample portfolio evaluation in the same way as in Section \ref{subsec:SP500}.
        The subpool of stocks with complete data in the preceding half year at the beginning of January, April, July, and October is selected. 
        The median and the maximum value of the numbers of stocks in the pool are 1626 and 1804,  respectively.
        The allocation is then determined using a 5-fold time-series cross-validation from the prior half year, and the assets are held until the first trade day of the next quarter. 

        Such a large portfolio problem is very challenging in practice.
        What is worse, data-driven methods based on historical data tend to produce highly sparse allocations, especially when facing the appearance of outliers.
        It happens more commonly in small-cap stocks, which fluctuate much more violently than large-cap stocks.
        To deal with this issue, we clip the stock relative prices within the top and bottom $2.5\%$ quantiles of the overall empirical distribution to exclude extreme values, which we find effective for all algorithms.
        We restrict the number of selected assets to be no less than a prespecified threshold $n_{\min}=100$.

        Their out-of-sample performances are again compared with benchmark portfolios as in Section \ref{subsec:benchmark}.
        As shown in Table \ref{tab:russell2000}, the MV methods fail to profit for the long run, and the EW and the GMV methods obtain good mean-variance efficiency. It coincides with \cite{demiguel2009optimal} that it is hard to beat the EW portfolio in the long run.
        Our methods are on par with GMV methods, while we only require half of the assets on average.
        Furthermore, our methods excel the GMV methods when the transaction fees are considered (\Cref{tab:russell2000_tranfee}) because the overall composition of the GMV portfolios between rebalancing periods is changing dramatically, indicating that choosing a suitable number of assets is beneficial.

        \begin{table}[!t]\footnotesize 
            \begin{tabularx}{\textwidth}{lPPPPP}
                \toprule
                \textbf{Method} & \textbf{Return} &  \textbf{Maximum Drawdown} & \textbf{Sharpe Ratio} & \textbf{Sortino Ratio} & \textbf{Avg. Num. of Assets}\\
                \midrule
                Benchmark &&&\\
                \hline
                \hline\addlinespace[0.2em]
                EW & 3.1023 & 0.6125 & 0.4176 & 0.5948 & 1640 \\
                GMV-P & 1.2725 & 0.2677 & 0.4122 & 0.5948 & 838 \\
                GMV-LS & 1.7997 & 0.3725 & 0.4402 & 0.6251 & 877 \\
                GMV-NLS & 1.5940 & 0.4502 & 0.3687 & 0.5247 & 820 \\
                MV-P & -0.7531 & 0.9398 & -0.0408 & -0.0554 & 302 \\
                MV-LS & -0.8415 & 0.9530 & -0.0879 & -0.1195 & 351 \\
                MV-NLS & -0.8987 & 0.9722 & -0.1180 & -0.1590 & 406 \\\midrule
                Our methods &&&\\
                \hline
                \hline\addlinespace[0.2em]
                LOG & 3.4483 & 0.5473 & 0.4514 & 0.6330 & 160 \\
                EXP-0.05 & 1.8137 & 0.4942 & 0.3738 & 0.5135 & 307 \\
                EXP-0.10 & 1.9293 & 0.4940 & 0.3866 & 0.5313 & 332 \\
                EXP-0.50 & 1.9791 & 0.4868 & 0.3953 & 0.5432 & 314 \\
                EXP-1.00 & 1.9514 & 0.4830 & 0.3929 & 0.5397 & 308 \\
                EXP-1.50 & 2.0456 & 0.4753 & 0.4076 & 0.5602 & 255 \\
                \bottomrule
            \end{tabularx}
            \caption{Out-of-sample results (without transaction fees) on Russell 2000 from 2005 to 2020.}\label{tab:russell2000}
        \end{table}
        
        We further explore augmenting portfolio performance with factor signals following \citep{ledoit2017nonlinear}. The factor signals are exposure to the stocks and may imply the research on efficient market as it enhances to explain cross-sectional stock returns. By now, many factors are documented in academia or constructed by the industry. We examine two simple factors, the moving average Sharpe ratio and the moving average relative strength index (RSI). The definitions of the two factors are detailed in Supplement Materials.
        The results are summarized in \Cref{tab:russell2000_factor_SR} and \Cref{tab:russell2000_factor_RSI}
        (without transaction fees) and \Cref{tab:russell2000_tranfee_factor_SR} and \Cref{tab:russell2000_tranfee_factor_RSI} (with transaction fees).
        We observed that the SRs and SoRs of the GMV portfolios considerably deteriorate because they are only concerned about minimizing the variation of the factor of the selected assets.
        On the contrary, the MV portfolios and our methods gain significant improvements compared with portfolios without factors.
        In most cases, our methods outperform the EW strategy containing 1,640 stocks on average in terms of all metrics.

        In addition to the portfolio performance, we are also interested in the time consumption of MV portfolios and our method when the asset number is much larger than the sample number. We choose the period when 1615 stocks are included in the training set of 128 samples. It takes about 60 seconds each fold each $\lambda_{\text{MV}}$ for MV-P to complete 1,000 iterations on average, while only 5 seconds each fold each $\lambda$ for LOG to complete 10,000 iterations, showing computational benefits of our methods in the ultra-high dimensional settings.
        
        \begin{table}[!t]\footnotesize 
            \begin{tabularx}{\textwidth}{lPPPPP}
                \toprule
                \textbf{Method} & \textbf{Return} &  \textbf{Maximum Drawdown} & \textbf{Sharpe Ratio} & \textbf{Sortino Ratio} & \textbf{Avg. Num. of Assets}\\
                \midrule
                Benchmark &&&\\
                \hline
                \hline\addlinespace[0.2em]
                EW & 3.1023 & 0.6125 & 0.4176 & 0.5948 & 1640 \\
                GMV-P & 0.1927 & 0.7264 & 0.1324 & 0.1925 & 813 \\
                GMV-LS & 0.6516 & 0.6640 & 0.2030 & 0.2906 & 806 \\
                GMV-NLS & 0.4707 & 0.6798 & 0.1755 & 0.2518 & 657 \\
                MV-P & 0.3664 & 0.5924 & 0.1728 & 0.2480 & 243 \\
                MV-LS & 0.6889 & 0.5837 & 0.2202 & 0.3150 & 275 \\
                MV-NLS & 0.5009 & 0.5765 & 0.1969 & 0.2819 & 330 \\ \midrule
                Our methods &&&\\
                \hline
                \hline\addlinespace[0.2em]
                LOG & 4.2410 & 0.5187 & 0.5009 & 0.7081 & 124 \\
                EXP-0.05 & 3.5765 & 0.5331 & 0.4678 & 0.6609 & 277 \\
                EXP-0.10 & 3.5342 & 0.5383 & 0.4656 & 0.6577 & 284 \\
                EXP-0.50 & 3.0583 & 0.5511 & 0.4382 & 0.6168 & 202 \\
                EXP-1.00 & 2.9818 & 0.5472 & 0.4332 & 0.6098 & 298 \\
                EXP-1.50 & 3.0385 & 0.5539 & 0.4365 & 0.6141 & 245 \\
                \bottomrule
            \end{tabularx}
            \caption{Out-of-sample results (without transaction fees) on Russell 2000 with SR factors.}\label{tab:russell2000_factor_SR}
        \end{table}
        
        \begin{figure}[t]
            \centering
              \includegraphics[width=\textwidth]{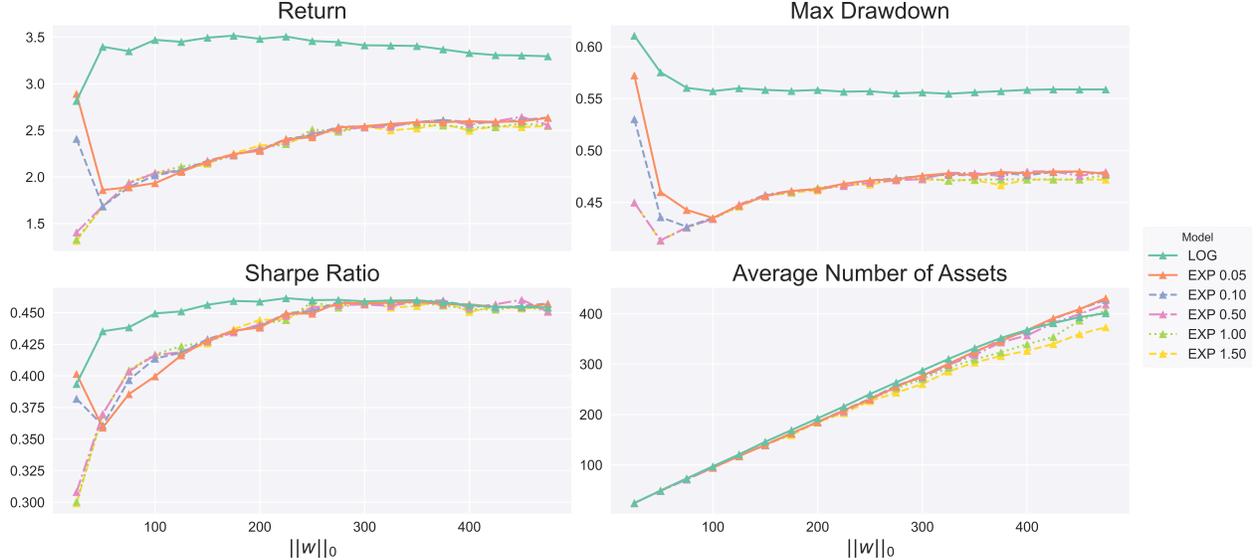}
            \caption{Portfolio performance on Russell 2000 historical data from 2005 to 2020 constrained on maximum number of assets, when $n=6$ months, $d\leq 1804$ and $n$\textsubscript{hold}$=3$ months.}
            \label{fig:russell2000_n}
        \end{figure}

        Finally, our algorithm is also flexible for producing portfolios with the desired number of holding assets.
        As shown in Figure \ref{fig:russell2000_n}, when imposing different cardinality constraints, the out-of-sample performances of our methods have similar patterns as seen in Section \ref{subsubsec:sp500_card}.
        For EXP portfolios, the aversion parameter also controls the maximum number of assets to hold.
        Overall, more assets help minimize risks but do not necessarily benefit more in terms of SRs and SoRs.
        Hence, different investors with different risk tolerance preferences can decide suitable choices of $s$ for their own based on historical data.

\section{Conclusion}\label{sec:conclusion}

    Strategies for high-dimensional portfolio management are becoming increasingly important as the global financial market expands. This portfolio management study provides new perspectives on the mean-variance efficiency, where an appropriate number of holding assets can facilitate a better trade-off between return and risk with more stability. 
    When the number of holding assets is chosen in a purely data-driven manner based on historical data, our portfolio strategies are on par with the best benchmark GMV portfolio strategies in Sharpe ratios and Sortino ratios, while only requiring less than half of assets as them, on both S\&P 500 and Russell 2000 datasets.
    The proposed strategies significantly outperform all benchmark methods when considering transaction costs.
    When the number of holding stocks is fixed in advance, we find that about 10$\sim$30 assets can maintain well-diversified portfolios on S\&P 500 datasets. At the same time, 160$\sim$300 are needed for Russell 2000 datasets, depending on the utility function one uses.
    To the best of our knowledge, this is the first case study concerning such a large number of assets in the literature on portfolio management.
    The results indicate that sparsity-induced portfolios will not ruin the diversification of the constrained portfolio with increasing risks. 
    For example, on Russell 2000, our best portfolio profits as much as the EW portfolio but reduces the maximum drawdown and the average number of assets by 10\% and 90\% respectively.
    Finally, the feasibility and stability of the proposed strategies with factor augmentation are verified.

    Focusing on the challenges in the ultrahigh-dimensional scenarios, we propose a general strategy based on the SAA formulation of EUM with cardinality constraints. 
    Our proposed strategy is effective even when the number of assets is far more than the sample size because it precludes the estimation error of the moments, a sore point of the mean-variance approach for large portfolios, and gives consistent solutions to the EUM problem. 
    The strategy is justified by the equivalence of the $l_1$-regularized problem on the nonnegative orthant and the $l_0$-constrained problem on probability simplex under mild conditions. 
    By working with this flexible convex optimization problem, we can identify a small subset of diversified assets with superior out-of-sample portfolio performance. Hence, we present a theoretically sound and computationally efficient strategy to make high-dimensional PM reliable and actionable in the rapidly growing financial asset market.

    Albeit the present study demonstrates successful examples in constructing diversified portfolios with our SAA-based method, the limitations of our study remain to be addressed. For example, almost all data-driven methods are susceptible to outliers when facing higher volatility. One possible direction is to incorporate robust optimization techniques into our analysis framework or use outlier detection algorithms to remove unwanted samples. Another limitation is that the current work considers long-only portfolios, while allowing short sales may be an important direction for further work. 
    Besides, we may introduce liquidity risk \citep{acharya2005asset}, or consider different types of investors.

\bibliographystyle{apalike}
\bibliography{ref}

\clearpage
\addcontentsline{toc}{chapter}{Supplementary material}
\begin{center}
{\large\bf SUPPLEMENTARY MATERIAL}
\end{center}

Supplementary Notes include the proof for all the theorems and extra experiment results.
The programs to replicate all our experiments can be obtained from \url{https://github.com/jaydu1/SparsePortfolio/tree/supplement}.

\renewcommand{\thesection}{\Alph{section}}
\beginsupplement
\section{Equivalence of \texorpdfstring{$l_0$}{l0} Constraints and \texorpdfstring{$l_1$}{l1} Regularizers}

    \subsection{Stochastic Dominance}
       Before we present the proofs of our theoretical results, we first introduce some tools in stochastic dominance that we would use in our proof.
        One way to compare the effects of different portfolios at a partial ordering of the available investments is to use the stochastic dominance rules.
        While there are different stochastic dominance rules, we shall use the second-order stochastic dominance (SDD) for our proof.
        \begin{definition}[Second-Order Stochastic Dominance]\label{def:SSD}
            A random variable $X$ second-order stochastic dominates the other random variable $Y$, denoted by $X\succeq_{SSD} Y$ if either one of the following conditions holds:
            \begin{enumerate}[(1)]
                \item $\EE[u(X)]\geq \EE[u(Y)]$ for all monotone increasing and concave utility functions $u$,
                \item $\EE[u(X)]\leq \EE[u(Y)]$ for all monotone decreasing and concave utility functions $u$,
                \item $\int_{-\infty}^z\overline{F}_{X}(t)\rd t\geq \int_{-\infty}^z\overline{F}_{Y}(t)\rd t$ for all $t\in\RR$,
            \end{enumerate}
            where $\overline{F}_{X}(t)=\PP(X>t)$ is the survival function for $X$.
        \end{definition}
        For the proof of the equivalence of the above two definitions, the readers may refer to \citep{levy2006stochastic}.
        Under some conditions, the single crossing property implies SSD, as shown in the following proposition.
        \begin{proposition}[Single Crossing and SSD]\label{prop:singlecross_SSD}
            Let $X$ and $Y$ be two random variables bounded below by $\eta$. If either $Y$ single crosses $X$ from below or they never cross in $\RR$, then the following statements are equivalent:
            \begin{enumerate}[(1)]
                \item $X\succeq_{SSD}Y$,
                \item $\EE[X]\geq \EE[Y]$.
            \end{enumerate}
        \end{proposition}
        \begin{proof}[Proof of Proposition \ref{prop:singlecross_SSD}]
            Our proof follows \citet{wolfstetter1993stochastic} but extends it to the unbounded case.
            From Definition \ref{def:SSD}, we have that (1) implies (2) with utility function $u(z)=z$.
            To prove the other side, the single crossing property implies that $\int_{-\infty}^z[\overline{F}_{X}(t)-\overline{F}_{Y}(t)]\rd t\geq 0$ for all $z\leq c$ where $c\in\RR$ is the crossing point.
            For $z>c$, notice that $X$ and $Y$ are bounded below by $\eta$, we have
            \begin{align*}
               0&\leq \EE[X+\eta]- \EE[Y+\eta]\\
               &\leq\int_{0}^{\infty}[\overline{F}_{X+\eta}(t)-\overline{F}_{Y+\eta}(t)]\rd t\\
               &\leq\int_{-\eta}^{\infty}[\overline{F}_{X}(t)-\overline{F}_{Y}(t)]\rd t\\
                &=\int_{-\eta}^z[\overline{F}_{X}(t)-\overline{F}_{Y}(t)]\rd t+\int_z^{\infty}[\overline{F}_{X}(t)-\overline{F}_{Y}(t)]\rd t\\
                &\leq \int_{-\eta}^z[\overline{F}_{X}(t)-\overline{F}_{Y}(t)]\rd t.
            \end{align*}
            Therefore, $\int_{-\infty}^z\overline{F}_{X}(t)\rd t\geq \int_{-\infty}^z\overline{F}_{Y}(t)\rd t$, $\forall\ z\in\RR$.
            On the other hand, if they never cross, then it naturally holds that $\int_{-\infty}^z\overline{F}_{X}(t)\rd t\geq \int_{-\infty}^z\overline{F}_{Y}(t)\rd t$, $\forall\ z\in\RR$.
            Therefore, $X\succeq_{SSD}Y$.
        \end{proof}

    \subsection{Proof for Expected Utility Maximization}
        \begin{proof}[Proof of Theorem \ref{thm:eq_l1_l0_E}]
            For any $\bw\in\RR^{d}_+\setminus\{\zero_d\}$, since $u$ is concave and $\cX_j$'s have a finite expectation, by Jensen's inequality we have $\EE[u(\bw^{\top}\cX)]\leq u(\EE[\bw^{\top}\cX])=u(\bw^{\top}\EE[\cX])<\infty$.
            On the other hand, $\EE[u(\bw^{\top}\cX)]\geq u(\eta_{\min}\|\bw\|_1)>-\infty$ since $u$ is nondecreasing. Thus, the expectation is well-defined.
            
            By the optimality conditions of the two optimization problems, we have
            \begin{align}
            \EE[u(Z_2)]&\geq \EE[u(\tilde{Z}_1)]\label{eq:opt_cond_1}\\ \EE[u(\tilde{Z}_2)]-\lambda\|\tilde{\bw}_2\|_1&\leq \EE[u(Z_1)]-\lambda\|\bw_1\|_1    \label{eq:opt_cond_2}
            \end{align}

            Note that if $\tilde{Z}_1$ and $Z_2$ are single-crossing, the property preserves for the one-to-one monotonic transformation of the two random variables.
            In particular, the crossing property preserves for $Z_1$ and $\tilde{Z}_2$.
            To see this, we observe that
            $\overline{F}_{{Z}_1}(t)=\overline{F}_{\tilde{Z}_1}(t\|\bw_1\|_1^{-1})$ and $\overline{F}_{\tilde{Z}_2}(t)=\overline{F}_{{Z}_2}(t|\bw_1\|_1^{-1})$.
            Analogously, the crossing property preserves for $u(Z_1)$, $u(\tilde{Z}_2)$ and $u(\tilde{Z}_1)$, $u({Z}_2)$.

            Next, we consider all the possible cases separably:

            (1) If $\tilde{Z}_1$ single crosses $Z_2$ from below or they do not cross, then $Z_2\succeq_{SSD}\tilde{Z}_1$ from Proposition \ref{prop:singlecross_SSD} since $\cX_j$'s are bounded below by $\eta_{\min}>0$.
            By applying the monotone increasing and concave transformation $z\mapsto u(\|\bw_1\|z)$ we have $\EE[u(\tilde{Z}_2)]\geq\EE[u(Z_1)]$ by Definition \ref{def:SSD}.
            Then $\EE[u(\tilde{Z}_2)]-\lambda\|\tilde{\bw}_2\|_1\geq\EE[u(Z_1)]-\lambda\|{\bw}_1\|_1$ since $\|\tilde{\bw}_2\|_1=\|\bw_1\|_1\cdot\|\bw_2\|_1=\|\bw_1\|_1$.
            With the optimality condition \eqref{eq:opt_cond_2}, we further have $\EE[u(\tilde{Z}_2)]-\lambda\|\tilde{\bw}_2\|_1=\EE[u(Z_1)]-\lambda\|\bw_1\|_1$ and $\EE[u(\tilde{Z}_2)]=\EE[u(Z_1)]$.
            Because $u$ is strictly concave, it has at most one maximizer, and the solution to problem \eqref{opt:l1_regularized} is unique.
            Then $\tilde{\bw}_2=\bw_1$ and $\bw_2=\tilde{\bw}_1$.

            (2) If $\tilde{Z}_1$ single crosses $Z_2$ from above, then $u(Z_1)$ single crosses $u(\tilde{Z}_2)$ from above.
            Repeating the similar argument as above, we have $\EE[u(Z_2)]\leq\EE[u(\tilde{Z}_1)]$.
            Combining this with the optimality condition \eqref{eq:opt_cond_1}, we have $\EE[u(Z_2)]=\EE[u(\tilde{Z}_1)]$.
            Also note that $\tilde{\bw}_1\in\Delta_d\cap\Sigma_d^s$.
            Therefore, $\tilde{\bw}_1$ is a solution to problem \eqref{opt:l0_constrained}.

            All in all, $\tilde{\bw}_1$ is a solution to problem \eqref{opt:l0_constrained}.
        \end{proof}

\section{Estimation Error}
    Based on Assumption \ref{ass:1}, we first established the cone condition and restricted strong convexity condition, which leads the proof for our main results.
    For ease of proof, we introduce some notations here.
    Recall $\hat{\bw}_n$ and $\bw^*$ are solutions for optimization problems \eqref{opt:M_estimator_pop} and \eqref{opt:M_estimator_emp} respectively.
    We define the error vector to be $\Delta=\hat{\bw}_n-\bw^*$, and the empirical first-order Taylor error as $\cE_n(\Delta)=h(\bw^*+\Delta)-h(\bw^*)-\langle\nabla_{\bw}h(\bw^*),\Delta\rangle$.
    Let $\BB_2(c_1,c_2)=\{\bv\mid c_1\leq\|\bv\|_2\leq c_2\}$ be the donut-shaped circle and $\BB_2(c)=\BB_2(0,c)$ be the ball.

    \subsection{Cone Condition}
        \begin{theorem}[Cone condition]
            Suppose that Assumption \ref{ass:1} (a), (c) and (d) hold.
            Let $\delta\in(0,1)$, $\lambda\geq 12\alpha(\nu\vee1)L_u\sigma\sqrt{\log(3/\delta)\log(2de/s)/n}$, and $\cS_0\subseteq[d]$ be the subset of indices of the $s$ entries of $\Delta$ that have largest magnitudes. It holds that $\Delta\in C(\cS_0,\gamma_1,\gamma_2) $ with probability at least $1-\delta/3$. 
            Here the unknown parameters $\alpha,\gamma_1,\gamma_2$ and the cone $C(\cS_0,\gamma_1,\gamma_2) $ are defined in Assumption \ref{ass:1} (d).  \label{thm:cone}
        \end{theorem}
        \begin{proof}[Proof of Theorem \ref{thm:cone}]
            By the optimality of $\hat{\bw}_n$, we have
            \begin{align*}
                h(\hat{\bw}_n)+\lambda\|\hat{\bw}_n\|_1\leq h(\bw^*)+\lambda\|\bw^*\|_1,
            \end{align*}
            which implies
            \begin{align}
                \lambda\|\bw^*\|_1-\lambda\|\hat{\bw}_n\|_1\geq h(\hat{\bw}_n)-h(\bw^*)\geq \langle \nabla_{\bw}h(\bw^*),\Delta\rangle,\label{eq:thm2_1}
            \end{align}
            since $h$ is convex. We next bound the magnitude of the above inner product. Note that $u'(\bX_i^{\top}\bw^*)(X_{ij}-\nu)$ is sub-Gaussian with parameter $L_u\sigma$, we have
            \begin{align}
                | \langle \nabla_{\bw}h(\bw^*),\Delta\rangle| & \leq \frac{1}{\sqrt{n}}\sum\limits_{j=1}^d\left|\frac{1}{\sqrt{n}}\sum\limits_{i=1}^nu'(\bX_i^{\top}\bw^*)X_{ij}\right|\cdot |\Delta_j|\notag\\
                & \leq \frac{1}{\sqrt{n}}\sum\limits_{j=1}^d\left(\left|\frac{1}{\sqrt{n}}\sum\limits_{i=1}^nu'(\bX_i^{\top}\bw^*)(X_{ij}-\nu)\right| + \nu\right) \cdot |\Delta_j|\notag\\
                &\leq 12(\nu\vee1)L_u\sigma\sqrt{\frac{\log(3/\delta)}{n}}\sum\limits_{j=1}^d\sqrt{\log(2d/j)}|\Delta_j|, \label{eq:thm2_2}
            \end{align}
            with probability at least $1-\delta/3$, where the last inequality is due to \citet[][Lemma 4]{dedieu2019improved} which uses a sorting technique to obtain upper bound for the maximum of sub-Gaussian sequences.
            By \eqref{eq:thm2_1} and \eqref{eq:thm2_2}, we have
            \begin{align}
                -12(\nu\vee1)L_u\sigma\sqrt{\frac{\log(3/\delta)}{n}}\sum\limits_{j=1}^d\sqrt{\log(2d/j)}|\Delta_j|\leq \lambda\|\bw^*\|_1-\lambda\|\hat{\bw}_n\|_1. \label{eq:thm2_3}
            \end{align}
            Without loss of generality, we assume that $|\Delta_1|\geq \cdots\geq |\Delta_d|$ and thus $\cS_0=[s]$. Let $\cS=\supp(\bw^*)$, we have
            \begin{align}
                \lambda\|\bw^*\|_1-\lambda\|\hat{\bw}_n\|_1&\leq \lambda\|\bw^*_{\cS}\|_1 - \lambda\|\hat{\bw}_{n\cS}\|_1 - \lambda\|\hat{\bw}_{n\cS^C}\|_1  \notag\\
                &\leq \lambda\|\Delta_{\cS}\|_1 - \lambda\|\Delta_{\cS^C}\|_1 \notag\\
                &\leq \lambda\|\Delta_{\cS_0}\|_1 - \lambda\|\Delta_{\cS_0^C}\|_1 \label{eq:thm2_4}
            \end{align}
            because $\cS_0$ is the set of indices whose entries have largest magnitudes.
            On the other hand, 
            \begin{align}
                &\ \ \ 12(\nu\vee1)L_u\sigma\sqrt{\frac{\log(3/\delta)}{n}}\sum\limits_{j=1}^d\sqrt{\log(2d/j)}|\Delta_j| \notag\\
                &= 12(\nu\vee1)L_u\sigma\sqrt{\frac{\log(3/\delta)}{n}}\left(\sum\limits_{j=1}^s\sqrt{\log(2d/j)}|\Delta_j|+ \sqrt{\log(2d/s)}\|\Delta_{\cS_0^C}\|_1\right)\notag\\
                &\leq 12(\nu\vee1)L_u\sigma\sqrt{\frac{\log(3/\delta)}{n}}\left(\sqrt{\sum\limits_{j=1}^s\log(2d/j)} \|\Delta_{\cS_0}\|_2+ \sqrt{\log(2d/s)}\|\Delta_{\cS_0^C}\|_1\right) \notag\\
                &\leq 12(\nu\vee1)L_u\sigma\sqrt{\frac{\log(3/\delta)}{n}}\left(\sqrt{s\log(2ed/s)} \|\Delta_{\cS_0}\|_2+ \sqrt{\log(2d/s)}\|\Delta_{\cS_0^C}\|_1\right) , \label{eq:thm2_4_5}
            \end{align}
            where the first inequality is due to Cauchy-Schwarz inequality and the last inequality is from Stirling formula.
            Let $\lambda\geq 12\alpha(\nu\vee1)L_u\sigma\sqrt{\log(3/\delta)\log(2de/s)/n}$, \eqref{eq:thm2_4_5} can be further upper bounded by
            \begin{align}
                12(\nu\vee1)L_u\sigma\sqrt{\frac{\log(3/\delta)}{n}}\sum\limits_{j=1}^d\sqrt{\log(2d/j)}|\Delta_j|\leq  \frac{\lambda}{\alpha}(\sqrt{s}\|\Delta_{\cS_0}\|_2+\|\Delta_{\cS_0^C}\|_1).\label{eq:thm2_5}
            \end{align}
            Combining \eqref{eq:thm2_3}, \eqref{eq:thm2_4}, and \eqref{eq:thm2_5}, we have that with probability at least $1-\delta/3$, $$-\frac{\lambda}{\alpha}(\sqrt{s}\|\Delta_{\cS_0}\|_2+\|\Delta_{\cS_0^C}\|_1) \leq \lambda\|\Delta_{\cS_0}\|_1 - \lambda\|\Delta_{\cS_0^C}\|_1,$$
            i.e.,
            $$\|\Delta_{\cS_0^C}\|_1 \leq \frac{\alpha}{\alpha-1}\|\Delta_{\cS_0}\|_1+ \frac{\sqrt{s}}{\alpha-1}\|\Delta_{\cS_0}\|_2,$$
            and the conclusion follows.
        \end{proof}
        
    \subsection{Restricted Strong Convexity}
        \begin{theorem}[Restricted strong convexity]
            Suppose that Assumption \ref{ass:1} (a)-(d) hold, and $\lambda\geq 12\alpha(\nu\vee1)L_u\sigma\sqrt{\log(3/\delta)\log(2de/s)/n}$, then
            $$\cE_n(\Delta) \geq \kappa\|\Delta\|_2^2 - 16L_u\|\Delta\|_1\sqrt{2\sigma^2\log(12d\log^2(d)/\delta)/n}, \quad\forall\ \Delta\in  \BB_2(1/d,2R)\cap C(\cS_0,\gamma_1,\gamma_2)$$
            with probability at least $1-\delta/3$.
            \label{thm:rsc}
        \end{theorem}
        \begin{proof}
            
            Since $u$ is $L_u$-Lipshitz continuous and $\|\Delta\|_2\leq \|\Delta\|_1\leq \|\bw^*\|_1+\|\hat{\bw}_n\|_1\leq 2R$, by \citet[Theorem 9.34]{wainwright2019high}, we have that
            $$|\cE_n(\Delta)-\EE[\cE_n(\Delta)]|\leq 16L_u\|\Delta\|_1\gamma,\quad\forall\ \Delta\in\BB_2(1/d,2R),$$
            with probability at least $1-4\log^2(d)\log(2R)\inf_{t>0}\EE[\exp(t(\|\overline{x}_n\|_{\infty}-\gamma))]$ where $\overline{x}_n=\sum_{i=1}^n\epsilon_i(X_i-\nu\one_d)/n$.
            Since $X_{ij}-\nu$ is sub-Gaussian with parameter $\sigma$, we have that $\sum_{i=1}^n\epsilon_i(X_{ij}-\nu)/n$ is sub-Gaussian with parameter $\sigma/\sqrt{n}$ and
            \begin{align*}
                \EE[\exp(t\|\overline{x}_n\|_{\infty})]& = \EE\left[\max\limits_{j\in[d]}\exp\left(t\sum_{i=1}^n\epsilon_i(X_{ij}-\nu)\right)\right]\\
                &\leq \sum\limits_{j=1}^d\EE\left[\exp\left(t\sum_{i=1}^n\epsilon_i(X_{ij}-\nu)\right)\right]\\
                &= d\exp\left(\frac{t^2\sigma^2}{2n}\right).
            \end{align*}
            Then
            \begin{align*}
                \inf\limits_{t>0}\EE[\exp(t(\|\overline{x}_n\|_{\infty}-\gamma))]&=\inf\limits_{t>0} d\exp\left(\frac{t^2\sigma^2}{2n}-t\gamma\right) = d\exp(- \frac{n\gamma}{2\sigma^2}).
            \end{align*}
            Let $4\log^2(d)d\exp(-n\gamma^2/(2\sigma^2))=\delta/3$, we have $\gamma=\sqrt{2\sigma^2\log(12d\log^2(d)/\delta)/n}$, and with probability at least $1-\delta/3$,
            \begin{align*}
                |\cE_n(\Delta)-\EE[\cE_n(\Delta)]|\leq 16L_u\sigma\|\Delta\|_1\sqrt{2\log(12d\log^2(d)/\delta)/n},\quad\forall\ \Delta\in\BB_2(1/d,2R).
            \end{align*}
            By Assumption \ref{ass:1} (d), $\EE[\cE_n(\Delta)]$ is locally $\kappa$-strongly convex on $C(\cS_0,\gamma_1,\gamma_2)$.
            Then with probability at least $1-\delta/3$, it holds that
            \begin{align*}
                \cE_n(\Delta)&\geq \EE[\cE_n(\Delta)] - 16L_u\sigma\|\Delta\|_1\sqrt{2\log(12d\log^2(d)/\delta)/n}\\
                &\geq \kappa\|\Delta\|_2^2 - 16L_u\sigma\|\Delta\|_1\sqrt{2\log(12d\log^2(d)/\delta)/n}
            \end{align*}
            for any $\Delta\in \BB_2(1/d,2R)\cap C(\cS_0,\gamma_1,\gamma_2)$.
        \end{proof}

    \subsection{Proof for Estimation Error}
        \begin{proof}[Proof of Theorem]
            Theorems \ref{thm:cone} and \ref{thm:rsc} yield that, with probability at least $1-2\delta/3$, $\Delta\in\BB_2(1/d,2R)\cap C(\cS_0,\gamma_1,\gamma_2)$ and
            \begin{align*}
               \kappa \|\Delta\|_2^2 &\leq \cE_n(\Delta) + 16L_u\sigma\|\Delta\|_1\sqrt{2\log(12d\log^2(d)/\delta)/n}\\
               &\leq \lambda\|\Delta_{\cS_0}\|_1 - \lambda\|\Delta_{\cS_0^C}\|_1 - \langle\nabla_{\bw}h(\bw^*),\Delta\rangle + 16L_u\sigma\|\Delta\|_1\sqrt{2\log(12d\log^2(d)/\delta)/n}\\
               &\leq \lambda\|\Delta_{\cS_0}\|_1 - \lambda\|\Delta_{\cS_0^C}\|_1 + (\|\nabla_{\bw}h(\bw^*)\|_{\infty}+16L_u\sigma\sqrt{2\log(12d\log^2(d)/\delta)/n})\|\Delta\|_1\\
               &\leq 2\lambda\|\Delta_{\cS_0}\|_1,
            \end{align*}
            when $\lambda\geq \|\nabla_{\bw}h(\bw^*)\|_{\infty}+16L_u\sigma\sqrt{2\log(12d\log^2(d)/\delta)/n}$.
            Next we bound $\|\nabla h(\bw^*)\|_{\infty}$ by sub-Gaussian concentration inequality:
            \begin{align*}
                \|\nabla_{\bw}h(\bw^*)\|_{\infty} &\leq \nu + \frac{1}{\sqrt{n}} \max\limits_{j\in[d]}\left|\frac{1}{\sqrt{n}}\sum\limits_{i=1}^n u'(\bX_i^{\top}\bw^*)(X_{ij}-\nu)\right|\\
                &\leq 2(\nu\vee 1)L_u\sigma\sqrt{\frac{2\log(3d/\delta)}{n}},
            \end{align*}
            with probability at least $1-\delta/3$.
            Then with probability at least $1-\delta$, the above probabilistic statements hold and we have
            \begin{align*}
                \kappa \|\Delta\|_2^2 
                &\leq  2\lambda\|\Delta_{\cS_0}\|_1\leq 2\lambda\sqrt{s}\|\Delta_{\cS_0}\|_2\leq 2\lambda\sqrt{s}\|\Delta\|_2,
            \end{align*}
            when
            \begin{align*}
                \lambda &\geq \max\{12\alpha(\nu\vee1)L_u\sigma\sqrt{\log(3/\delta)\log(2de/s)/n},\\
                &\qquad\qquad2(\nu\vee 1)L_u\sigma\sqrt{2\log(3d/\delta)/n} + 16L_u\sigma\sqrt{2\log(12d\log^2(d)/\delta)/n}\}.
            \end{align*}
            It suffices to set $ \lambda\gtrsim \alpha (\nu\vee 1)L_u\sigma \sqrt{\log(d/\delta)/n}$.
            Thus, with probability at least $1-\delta$, it holds that
            \begin{align*}
                \|\Delta\|_2 & \lesssim \cO(\kappa^{-1}(\nu\vee 1)L_u\sigma\sqrt{s\log(d/\delta)/n}).
            \end{align*}
        \end{proof}

\section{Screening}

    \begin{proof}[Proof of Theorem \ref{thm:primal_dual}]
        By introducing the new variables $\bZ=\bX\bw$, the Lagrange function of the primal problem is given by
        \begin{align*}
            L(\bw,\bZ,\btheta)&=H(\bZ)+\lambda\|\bw\|_1+\mathds{1}_{\RR^d_+}(\bw)+\lambda \btheta^{\top}(\bZ-\bX\bw),
        \end{align*}
        and the dual function is given by
        \begin{align*}
            \cD_{\lambda}(\btheta)&=\inf_{\bw,\bZ}L(\bw,\bZ,\btheta)\\
            &= \inf_{\bZ}\ (\lambda \btheta^{\top}\bZ + H(\bZ)) +  \lambda\inf_{\bw}\ (\|\bw\|_1+\mathds{1}_{\RR^d_+}(\bw)-\btheta^{\top}\bX\bw) \\
            &= - H^*(-\lambda\btheta) - \lambda (\|\cdot\|_1+\mathds{1}_{\RR^d_+}(\cdot))^*(\bX^{\top}\btheta).
        \end{align*}
        From the scaling property of convex conjugate $[af(x)]^*=af^*(x^*/a)$, we have that
        \begin{align*}
            H^*(-\lambda\btheta)&=-
            \frac{1}{n}\sum\limits_{j=1}^nu^*(n\lambda \theta_j).
        \end{align*}
        On the other hand, the convex conjugate of the $l_1$ norm is the indicator function of the unit ball of its dual norm, which implies that
        \begin{align*}
            (\|\cdot\|_1+\mathds{1}_{\RR^d_+}(\cdot))^*(\xi)&=
            \inf_{\xi_1+\xi_2=\xi}\mathds{1}_{\{\xi'\in\RR^n\mid  \|\xi'\|_{\infty}\leq 1\}}(\xi_1)+\mathds{1}_{\RR^d_-}(\xi_2)
            =\mathds{1}_{\{\xi'\in\RR^n\mid  \|\phi(\xi')\|_{\infty}\leq 1\}}(\xi),
        \end{align*}
        where the first equality comes from
        the property that the conjugate of a sum is the infimal convolution of the individual conjugates.
        Thus, we have
        \begin{align*}
            \cD_{\lambda}(\btheta)&=\frac{1}{n}\sum\limits_{j=1}^nu^*(n\lambda \theta_j),\qquad \btheta\in \cC_D=\dom(\cD_{\lambda})\cap \{\btheta\in\RR^n\mid  \|\phi(\bX^{\top}\btheta)\|_{\infty}\leq 1\}.
        \end{align*}
        The first-order optimality conditions \citep[][Proposition 19.18]{bauschke2011convex} give rise to \eqref{eq:primal_dual_link} and
        $\bX^{\top}_{1:n} \btheta^* \in \partial \|\bw^*\|_1$, which is equivalent to \eqref{eq:subdiff}.
    \end{proof}

    \begin{proof}[Proof of Theorem \ref{thm:safe_region}]
        Since $\cD_{\lambda}$ is $\alpha$-strongly concave, we have that
        \begin{align*}
            \cD_{\lambda}(\btheta^*)\leq \cD_{\lambda}(\btheta)+\langle \nabla\cD_{\lambda}(\btheta),\btheta^*-\btheta\rangle -\frac{\alpha}{2}\|\btheta^*-\btheta\|_2^2.
        \end{align*}
        By the optimality of $\btheta^*$, $\langle \nabla\cD_{\lambda}(\btheta),\btheta^*-\btheta\rangle\leq 0$.
        This implies that
        \begin{align*}
            \cD_{\lambda}(\btheta^*)\leq \cD_{\lambda}(\btheta) -\frac{\alpha}{2}\|\btheta^*-\btheta\|_2^2.
        \end{align*}
        The weak duality gives that $\cD_{\lambda}(\btheta)\leq \cD_{\lambda}(\btheta^*)\leq \cP_{\lambda}(\bw^*)\leq \cP_{\lambda}(\bw)$ for all $(\bw,\btheta)\in\cC_P\times\cC_D$.
        Then, $\alpha/2\cdot\|\btheta^*-\btheta\|_2^2\leq \cP_{\lambda}(\bw^*)-\cD_{\lambda}(\btheta^*)\leq \cP_{\lambda}(\bw)-\cD_{\lambda}(\btheta)=\Gap_{\lambda}(\bw,\btheta)$ and the inclusion of the optimal dual variable $\btheta^*$ in the safe region follows.

        Note that
        \begin{align*}
            \max\limits_{\btheta'\in \cB(\btheta, r) }\phi(\bX_{\cdot j}^{\top}\btheta')&=\max\limits_{\btheta'\in \cB(\btheta, r) }\phi(\bX_{\cdot j}^{\top}(\btheta+\btheta'-\btheta))\\
            &\leq \phi(\bX_{\cdot j}^{\top}\btheta)+r\max\limits_{\bu\neq \zero }\frac{|\bX_{\cdot j}^{\top}\bu|}{\|\bu\|_2}\\
            &\leq \phi(\bX_{\cdot j}^{\top}\btheta)+r\|\bX_{\cdot j}\|_{2},
        \end{align*}
        where the first inequality is due to the triangle inequality, the second inequality comes from Cauchy-Schwartz inequality.
        Thus we have
        \begin{align*}
          \phi(\bX_{\cdot j}^{\top}\btheta)+r\|\bX_{\cdot j}\|_{2}<1
            &\quad\Longrightarrow\quad
            \max\limits_{\btheta'\in \cB(\btheta, r) }\phi(\bX_{\cdot j}^{\top}\btheta')<1\\
            &\quad\Longrightarrow\quad |\bX_{\cdot j}^{\top}\btheta^*|<1 \quad\Longrightarrow\quad w^*_j=0.
        \end{align*}
    \end{proof}

    \begin{proof}[Proof of Theorem \ref{thm:convergence}]
        Recall that at each iteration
        \begin{align}
            \bw^{(t+1)}= \Prox_{\iota\lambda \Omega}(\bw^{(t)} - \iota \nabla_{\bw}h(\bw^{(t)})), \label{eq:prox_update}
        \end{align}
        where $\Omega(\bw)=\|\bw\|_1+\mathds{1}_{\RR^d_+}(\bw)$.
        Define the generalized gradient as
        \begin{align*}
            G_\iota(\bw')=\frac{1}{\iota}(\bw' - \Prox_{\iota\lambda\Omega}(\bw' - \iota \nabla_{\bw}h(\bw'))
        \end{align*}
        so that the proximal update \eqref{eq:prox_update} can be rewritten as
        \begin{align}
            \bw^{(t+1)}= \bw^{(t)} - \iota G_\iota(\bw^{(t)}), \label{eq:proximal_GD}
        \end{align}
        and the optimality condition implies $G_\iota(\bw^*)=\zero$.

        By Assumption \ref{ass:conv_smooth}, we have that
        \begin{align}
            h(\bw^{(t+1)})\leq h(\bw^{(t)})+\langle\nabla_{\bw}h(\bw^{(t)}),\bw^{(t+1)}-\bw^{(t)}\rangle + \frac{L_{\nabla u}}{2}\|\bw^{(t+1)}-\bw^{(t)}\|_2^2. \label{eq:smooth}
        \end{align}
        By the convexity of $h$, we have that for any $\bv$,
        \begin{align}
            h(\bv)\geq h(\bw^{(t)})+\langle\nabla_{\bw}h(\bw^{(t)}),\bv-\bw^{(t)}\rangle \label{eq:convexity}
        \end{align}
        Combining \eqref{eq:smooth} with \eqref{eq:convexity},
        \begin{align}
            h(\bw^{(t+1)})&\leq h(\bv) - \langle\nabla_{\bw}h(\bw^{(t)}),\bv-\bw^{(t)}\rangle + \langle\nabla_{\bw}h(\bw^{(t)}),\bw^{(t+1)}-\bw^{(t)}\rangle + \frac{L_{\nabla u}}{2}\|\bw^{(t+1)}-\bw^{(t)}\|_2^2\notag\\
            &=  h(\bv) + \langle\nabla_{\bw}h(\bw^{(t)}),\bw^{(t+1)}-\bv\rangle + \frac{L_{\nabla u}}{2}\|\bw^{(t+1)}-\bw^{(t)}\|_2^2.
            \label{eq:smooth_conv}
        \end{align}

        Since $\Omega$ is convex, we have that for any $g\in \partial \Omega(\bw^{(t)})$,
        \begin{align}
            \Omega(\bv)\geq \Omega(\bw^{(t)}) + \langle g, \bv-\bw^{(t)}\rangle. \label{eq:subgrad1}
        \end{align}
        In the $t$th proximal update, the optimality condition implies that
        \begin{align}
            \zero_d\in \partial \Omega(\bw^{(t+1)}) + \frac{1}{\iota} (\bw^{(t+1)}- \bw^{(t)} + \iota \nabla_{\bw}h(\bw^{(t)})),
        \end{align}
        which is equivalent to
        \begin{align}
            G_{\iota}(\bw^{(t)})- \nabla_{\bw}h(\bw^{(t)})\in \partial \Omega(\bw^{(t)}). \label{eq:subgrad2}
        \end{align}
        Combining \eqref{eq:smooth_conv}, \eqref{eq:subgrad1} and \eqref{eq:subgrad2}, we have
        \begin{align}
            h(\bw^{(t+1)})+\lambda \Omega(\bw^{(t+1)})&\leq h(\bv) + \lambda \Omega(\bv)
            + \langle G_{\iota}(\bw^{(t)}), \bw^{(t+1)}-\bv\rangle
            + \frac{L_{\nabla u}}{2}\|\bw^{(t+1)}-\bw^{(t)}\|_2^2.
             \notag
        \end{align}
        Because the step size $\iota\leq \frac{1}{L_u}$ and from the update rule \eqref{eq:proximal_GD},
        \begin{align}
            \cP_{\lambda}(\bw^{(t+1)})&\leq \cP_{\lambda}(\bv) + \langle G_{\iota}(\bw^{(t)}), \bw^{(t)}-\iota G_{\iota}(\bw^{(t)})-\bv\rangle
            + \frac{L_{\nabla u}}{2}\|\bw^{(t+1)}-\bw^{(t)}\|_2^2\notag\\
            &\leq \cP_{\lambda}(\bv) + \langle G_{\iota}(\bw^{(t)}), \bw^{(t)}-\bv\rangle
            - \frac{1}{2\iota}\|\bw^{(t+1)}-\bw^{(t)}\|_2^2 \label{eq:primal_central_lemma}
        \end{align}
        By setting $\bv=\bw^{(t)}$ and $\bv=\bw^*$ in \eqref{eq:primal_central_lemma}, we have that
        \begin{align*}
            \cP_{\lambda}(\bw^{(t+1)})&\leq\cP_{\lambda}(\bw^{(t)})
            - \frac{1}{2\iota}\|\bw^{(t+1)}-\bw^{(t)}\|_2^2\\
             \cP_{\lambda}(\bw^{(t+1)})&\leq \cP_{\lambda}(\bw^*) + \langle G_{\iota}(\bw^{(t)}), \bw^{(t)}-\bw^*\rangle
            - \frac{1}{2\iota}\|\bw^{(t+1)}-\bw^{(t)}\|_2^2.
        \end{align*}
        Thus,
        \begin{align*}
            \cP_{\lambda}(\bw^{(t+1)})-\cP_{\lambda}(\bw^*)&\leq \langle G_{\iota}(\bw^{(t)}), \bw^{(t)}-\bw^*\rangle
            - \frac{1}{2\iota}\|\bw^{(t+1)}-\bw^{(t)}\|_2^2\\
            &=\frac{1}{2\iota} (\langle 2\iota G_{\iota}(\bw^{(t)}), \bw^{(t)}-\bw^*\rangle - \|\iota G_{\iota}(\bw^{(t)})\|_2^2\\
            &= \frac{1}{2\iota} (\|\bw^{(t)}-\bw^*\|_2^2 - \|\iota G_{\iota}(\bw^{(t)}) - \bw^{(t)}+\bw^*\|_2^2)\\
            &= \frac{1}{2\iota} (\|\bw^{(t)}-\bw^*\|_2^2 - \|\bw^{(t+1)}-\bw^*\|_2^2).
        \end{align*}
        By summing up the above equation for $t=1,\ldots,\tau$, we get
        \begin{align*}
           (\tau+1)( \cP_{\lambda}(\bw^{(t+1)})-\cP_{\lambda}(\bw^*))&\leq \frac{1}{2\iota} (\|\bw^{(0)}-\bw^*\|_2^2 - \|\bw^{(\tau+1)}-\bw^*\|_2^2)\\
           &\leq \frac{1}{2\iota} \|\bw^{(0)}-\bw^*\|_2^2,
        \end{align*}
        which gives a convergence rate when no screening is performed.


        By Theorem \ref{thm:safe_region}, $\forall\ t\in\NN_+$, $\forall\ j\in\cS^{(t)}\setminus\cS^{(t-1)}$, $w_j^*=0$. That is, the inactive features are safely screened out.
        In other words, the convergence is still guaranteed as the screened features are exactly zeros.
    \end{proof}

\section{Proximal Gradient Descent}\label{supp:PGD}
    Recall that the proximal operator of any closed proper convex function $r:\RR^d\rightarrow\RR\cup\{+\infty\}$ is defined as
    \begin{align*}
        \Prox_{r}(\bw)=\argmin_{\bv\in\RR^d} \frac{1}{2}\|\bw-\bv\|_2^2 + r(\bv).
    \end{align*}
    Therefore, the proximal operator defined in Algorithm \ref{algo:spo} is given by
    \begin{align}
        \Prox_{\iota\lambda \Omega}(\bw)=\argmin_{\bv\in\RR^d} \frac{1}{2}\|\bw-\bv\|_2^2 + \iota\lambda\|\bv\|_1 + \mathds{1}_{\RR^d_+}(\bv). \label{eq:prox}
    \end{align}
    To evaluate this proximal operator, we need to solve the optimization problem \eqref{eq:prox}.
    As the proximal operator of $r_1(\bv)= \iota\lambda\|\bv\|_1$ and $r_2(\bv)=\frac{1}{2}\|\bw-\bv\|_2^2+\mathds{1}_{\RR^d_+}(\bv)$ can be efficiently evaluated, we can use the \emph{alternating direction method of multipliers (ADMM)} to compute \eqref{eq:prox}.
    More specifically, we perform the following iterations until $\bv^{(k)}$ converges:
    \begin{align}
        \bv^{(k+1)} &= \Prox_{r_1}(\bp^{(k)} - \bq^{(k)})= \cS\cT_{\frac{\lambda\iota}{\rho}}(\bp^{(k)} - \bq^{(k)}) \label{eq:ADMM_1}\\
        \bp^{(k+1)} &= \Prox_{r_2}(\bv^{(k+1)} + \bq^{(k)}) = \frac{1}{1+\rho} \phi(\bw + \rho(\bv^{(k+1)} + \bq^{(k)}))\label{eq:ADMM_2}\\
        \bq^{(k+1)}&=\bq^{(k)}+\bv^{(k+1)}-\bp^{(k+1)},\label{eq:ADMM_3}
    \end{align}
    where $\cS\cT_t$ is the soft-thresholding operator with threshold $t$, $\rho>0$ is the step size parameter and usually set to one.
    The initial values are given by $\bp^{(0)}=\phi(\bw)$ and $\bq^{(0)}=\bw-\phi(\bw)$.

    If $w_j$ is already in $\RR^d_+$, then only one step of soft thresholding is needed to evaluate the $j$th coordinate of \eqref{eq:prox} from the update rules \eqref{eq:ADMM_1}-\eqref{eq:ADMM_3},
    Otherwise, the ADMM updates will bring $[\Prox_{\lambda\iota \Omega}(\bw)]_j$ closer to zero than $w_j$.
    To see this, suppose that for $t\leq k$, $\bq^{(t)}\preceq \zero$ and $\bp^{(t)} - \bq^{(t)}\preceq \bw$. Then for $t=k+1$,
    \begin{align*}
        \bp^{(k+1)}-\bq^{(k+1)}&=\bq^{(k)}+\bv^{(k+1)}
        = \bq^{(k)} + \cS\cT_{\frac{\lambda\iota}{\rho}}(\bp^{(k)} - \bq^{(k)})
        \preceq \bw.
    \end{align*}
    By induction, we have that $\bv^{(k+1)}\preceq\cS\cT_{\frac{\lambda\iota}{\rho}}(\bw)$ for $k=0,1,2,\ldots$.
    Therefore, we can escape the ADMM algorithm with a smaller value of $w_j$ and continue the iterations in Algorithm \ref{algo:spo}.
    Furthermore, as the objective function is scale-invariant monotonic, it suffices to utilize root-finding methods such as the bisection method at line 5 in Algorithm \ref{algo:spo} to accelerate the optimization.
    We propose an acceleration version of PGD that incorporates bisection steps and exploration steps in Algorithm \ref{algo:acceleration}.

    \begin{table}[!h]\footnotesize
        \begin{tabularx}{0.97\textwidth}{ccYY}
            \toprule[.1em]\addlinespace[0.5em]
            utility function & $u(z)$ &  $1-\exp(-(a z+\eta)),\ a,\eta>0$ & $\log(z+\eta),\ \eta>0$ \\\addlinespace[0.5em]\cline{2-4}\addlinespace[0.5em]
            \multirow{2}{*}{empirical function} & $H(\bz)$  &
            \multicolumn{2}{c}{$-\frac{1}{n}\sum_{j=1}^nu(z_j)$}
            \\\addlinespace[0.5em]\cline{3-4}\addlinespace[0.5em]
            & $h(\bw)$ & \multicolumn{2}{c}{$H(\bX\bw)$} \\\addlinespace[0.5em]\cline{2-4}\addlinespace[0.5em]
            \multirow{2}{*}{gradient} & $\nabla_{\bz}H(\bz)$ & $-\frac{a}{n}
            \left[\begin{smallmatrix}
              \exp(-(a z_1+\eta)) & \cdots & \exp(-(a z_n+\eta))
            \end{smallmatrix}\right]^{\top}$  & $-\frac{1}{n}
            \left[\begin{smallmatrix}
              \frac{1}{z_1+\eta}  & \cdots & \frac{1}{z_n+\eta}
            \end{smallmatrix}\right]^{\top}$  \\\addlinespace[0.5em]\cline{3-4}\addlinespace[0.5em]
             & $\nabla_{\bw}h(\bw)$ & \multicolumn{2}{c}{$\bX^{\top}\nabla_{\bz}H(\bX\bw)$} \\\addlinespace[0.5em]\cline{2-4}\addlinespace[0.5em]
             \multirow{4}{*}{Lipschitz constant}
             & $L_u$ &  $a\exp(-\eta)$ & $\frac{1}{\eta}$ \\\addlinespace[0.5em]\cline{3-4}\addlinespace[0.5em]
             & $L_{\nabla H}$ & $\frac{a^2}{n\exp(\eta)}$  & $\frac{1}{n\eta^2}$  \\\addlinespace[0.5em]\cline{3-4}\addlinespace[0.5em]
             & $L_{\nabla h}$ &
             \multicolumn{2}{c}{$\|\bX^{\top} \nabla_{\bz}^2H(\zero_n) \bX\|_2$} \\\addlinespace[0.5em]\cline{2-4}\addlinespace[0.5em]
            primal function & $P(\bw)$ & \multicolumn{2}{c}{$h(\bw)+\lambda\|\bw\|_1$} \\\addlinespace[0.5em]\cline{2-4}\addlinespace[0.5em]
            \multirow{6}{*}{dual function} & $u^*(\theta)$ & $-\frac{\theta}{a}\log\left(\frac{\theta}{a}\right)+\frac{\theta}{a}-\theta\eta-1$ & $\log(\theta)-\eta\theta+1$\\\addlinespace[0.5em]\cline{3-4}\addlinespace[0.5em]
            & $\cD(\btheta)$  & \multicolumn{2}{c}{$\frac{1}{n}\sum_{j=1}^nu^*(n\lambda\theta_j)$} \\\addlinespace[0.5em]\cline{3-4}\addlinespace[0.5em]
            & $\dom(\cD)$  & \multicolumn{2}{c}{$\RR^n_{+}$} \\\addlinespace[0.5em]\cline{3-4}\addlinespace[0.5em]
            & $\nabla^2 \cD(\btheta)$  & $\diag\left(-\frac{\lambda}{a\theta_1},\cdots,-\frac{\lambda}{a\theta_n}\right)$ &  $\diag\left(-\frac{1}{n\theta_1^2},\cdots,-\frac{1}{n\theta_n^2}\right)$\\
            \midrule[.1em]
        \end{tabularx}
        \caption{Useful ingredients for performing the Algorithm \ref{algo:spo}.}\label{tab:ingredients}
    \end{table}
    \renewcommand{\thealgorithm}{\Alph{section}.\arabic{algorithm}}
    \begin{algorithm}[!h]
      \caption{Acceleration for PGD}\label{algo:acceleration}\small
      \begin{algorithmic}[1]
        \REQUIRE The allocation vector $\bw^{(t-1)}$ and $\bw^{(t)}$, the gradient $\bg=\nabla_{\bw}h(\bw)$, the set of active coordinates $\cA$, the regularization parameter $\lambda$, and the step size $\iota$.
        \STATE Initialize $\bw^{(t+1)}=\zero_d$.
        \FOR{$j$ in $\cA$}
            \STATE $w_{j}=\max\{w_j^{(t)} - \iota(g_j +\lambda) ,0\}$.
            \IF{$(w_j-w_j^{(t)})\cdot(w_j^{(t)}-w_j^{(t-1)})<0$}
                \IF{$w_j-w_j^{(t)}>0$}
                    \STATE $w_j^{(t+1)}=\frac{1}{2}(w_j^{(t)}+\min\{w_j^{(t-1)},w_j\})$.
                \ELSE{}
                    \STATE $w_j^{(t+1)}=\frac{1}{2}(w_j^{(t)}+\max\{w_j^{(t-1)},w_j\})$.
                \ENDIF
            \ELSIF{$w_{j}>w_j^{(t)}>w_j^{(t-1)}$}
                \STATE $w_j^{(t+1)}= 2 w_j- w_j^{(t)}$.
            \ELSE
                \STATE $w_j^{(t+1)}=w_j$.
            \ENDIF
        \ENDFOR
        \ENSURE The updated allocation $\bw^{(t+1)}$.
      \end{algorithmic}
    \end{algorithm}

\clearpage
\section{Comparison with SCS}

    The general convex optimization solvers in \textsc{Cvxpy} \citep{agrawal2018rewriting} will first transform optimization problem \eqref{opt:l1_regularized} into standard form, and then solve the transformed problem instead.
    The transformed problem contains slack variables that lie in the intersect of a nonnegative cone and an exponential cone. 
    Available in \textsc{Cvxpy} \citep{agrawal2018rewriting}, Splitting Conic Solver (SCS) \citep{scs} is the only open-source solver that is capable to work with nonnegative cone and exponential cone constraints at the same time.
    Hence, we compare the accuracy and time consumption of our proposed algorithm with SCS on solving problem \eqref{opt:l1_regularized}.
    
    Since different methods use different stopping criteria, we adopt the following procedure to evaluate the accuracy and efficiency of the two methods.
    Firstly, we ran SCS with maximum number of iteration $1,000$ and stopping tolerance $10^{-16}$ and recorded the actual run time.
    Then our proposed algorithm was applied to the same data with the same stopping tolerance and was forced to stop if the time reached SCS' run time. By controlling the run time, we are able to compare with the accuracy of both methods.
    Secondly, since SCS applies Douglas-Rachford splitting and update the variables $\bw^{(t)}$ iteratively as our proposed algorithm, we ran both methods for 1,000 iterations with stopping tolerance $10^{-16}$ and compared the time consumption.
    We repeat this process on data with varying dimension $d$ sampled from NYSE dataset in Section \ref{subsec:NYSE}.
    The results are shown in Figure \ref{fig:cvx_comparison}.

    Note that the objective value with zero portfolio is zero and the objective values returned by two methods are all negative.
    Thus, a larger ratio will indicate better performance of the proposed method relative to SCS.
    We observed that the two methods are comparable in terms of accuracy for the logarithm utility, as the ratios of objective values are near one.
    However, there is much difference for exponential utility especially when the risk aversion parameter is small.
    In terms of execution time for the same number of iterations, we see that the proposed algorithm significantly outperforms SCS.
    Its dependence on dimension $d$ grows much slower than the one of SCS.
    This is because useless features are safely screened out and thus we can save much time by skipping these non-active features for the following iterations.
    Such benefit is more pronounced when the dimension $d$ is large.
    
    \begin{figure}[t]
        \centering
        \includegraphics[width=0.9\textwidth]{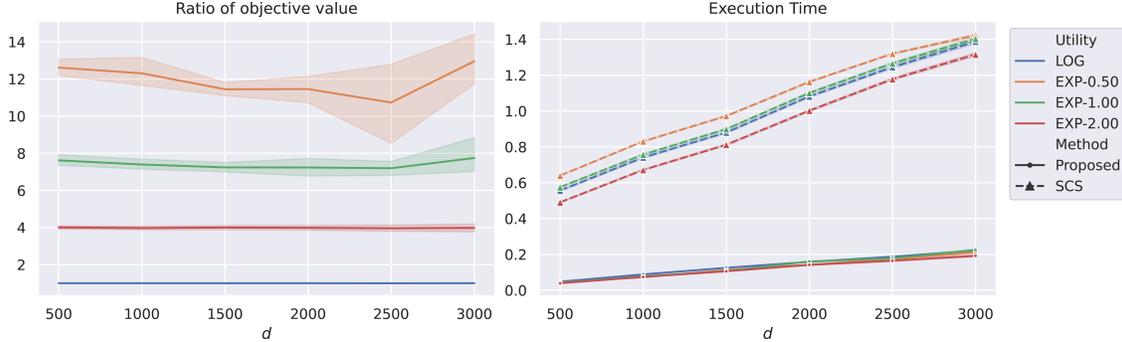}
        \caption{Comparison of the proposed algorithm and SCS with respect to dimension $d$. The left panel shows the ratio of objective value between the proposed method and SCS within the same time; the right panel shows execution time of the two methods for 1,000 iterations.
        The shaded regions represent values within one standard deviation of the mean across 500 experiments (data generated from 10 random seeds and a path of 50 $\lambda$'s).}
        \label{fig:cvx_comparison}
    \end{figure}
    
    \ 

\section{Additional Experiment Results}
    \subsection{Portfolio without factors}
    
    \begin{table}[!h]\small
        \begin{tabularx}{\textwidth}{lPPPPP}
            \toprule
            \textbf{Method} & \textbf{Return} &  \textbf{Maximum Drawdown} & \textbf{Sharpe Ratio} & \textbf{Sortino Ratio} & \textbf{Avg. Num. of Assets}\\
            \midrule
            Benchmark &&&\\
            \hline
            \hline\addlinespace[0.2em]
            EW & 0.3844 & 0.6544 & 0.1665 & 0.2346 & 1640 \\
            GMV-P & -0.0563 & 0.3515 & -0.0810 & -0.1123 & 838 \\
            GMV-LS & 0.1247 & 0.4385 & 0.0392 & 0.0544 & 877 \\
            GMV-NLS & 0.1770 & 0.4903 & 0.0729 & 0.1021 & 820 \\
            MV-P & -0.8434 & 0.9471 & -0.1112 & -0.1508 & 302 \\
            MV-LS & -0.9057 & 0.9570 & -0.1649 & -0.2237 & 351 \\
            MV-NLS & -0.9447 & 0.9761 & -0.2030 & -0.2728 & 406 \\\midrule
            Our methods &&&\\
            \hline
            \hline\addlinespace[0.2em]
            LOG & 2.4004 & 0.5558 & 0.3839 & 0.5369 & 160 \\
            EXP-0.05 & 0.8602 & 0.5101 & 0.2365 & 0.3231 & 307 \\
            EXP-0.10 & 0.8946 & 0.5099 & 0.2424 & 0.3312 & 332 \\
            EXP-0.50 & 1.0029 & 0.5033 & 0.2620 & 0.3582 & 314 \\
            EXP-1.00 & 0.9931 & 0.5005 & 0.2606 & 0.3562 & 308 \\
            EXP-1.50 & 1.1702 & 0.4885 & 0.2912 & 0.3985 & 255 \\
            \bottomrule
        \end{tabularx}
        \caption{Out-of-sample results (with transaction fees) on Russell 2000 from 2005 to 2020.}\label{tab:russell2000_tranfee}
    \end{table}

    \subsection{Portfolio with factors}
    As discussed in Section 4.3, we turn attention to our methods with factor signals, which thereby augment the empirical performances. For simplicity and reproducibility, we use two simple factor signals and the methods are readily extended for incoporating multi-factors. Two factors are listed below:
    \begin{itemize}
        \item SR: Moving average statistics of the in-sample Sharpe ratios with window size of 26.
        \item RSI: Relative strength index, whose value at time $t$ is 
    formally defined through the following equations:
    \begin{align*}
        DIF_t &= close_t-open_t\\
        RSI_t &= \frac{\mathrm{ewm}\left(\{DIF_j\vee 0\}_{j\leq t},\frac{1}{24}\right)}{\mathrm{ewm}\left(\{|DIF_j|\}_{j\leq t},\frac{1}{24}\right)}
    \end{align*}
    where $\mathrm{ewm}(\cdot,\alpha)$ denote the exponentially weighted operation with smoothing factor $\alpha$.
    \end{itemize}
    
       \begin{table}[!h]\small
        \begin{tabularx}{\textwidth}{lPPPPP}
            \toprule
            \textbf{Method} & \textbf{Return} &  \textbf{Maximum Drawdown} & \textbf{Sharpe Ratio} & \textbf{Sortino Ratio} & \textbf{Avg. Num. of Assets}\\
            \midrule
            Benchmark &&&\\
            \hline
            \hline\addlinespace[0.2em]
            EW & 0.3844 & 0.6544 & 0.1665 & 0.2346 & 1640 \\
            GMV-P & -0.5015 & 0.7857 & -0.0651 & -0.0936 & 813 \\
            GMV-LS & -0.3340 & 0.7212 & -0.0183 & -0.0259 & 806 \\
            GMV-NLS & -0.3001 & 0.7155 & -0.0032 & -0.0045 & 657 \\
            MV-P & -0.0804 & 0.6540 & 0.0894 & 0.1278 & 243 \\
            MV-LS & 0.1105 & 0.5853 & 0.1344 & 0.1916 & 275 \\
            MV-NLS & -0.0850 & 0.5971 & 0.0966 & 0.1375 & 330 \\
            \midrule
            Our methods &&&\\
            \hline
            \hline\addlinespace[0.2em]
            LOG & 3.1726 & 0.5266 & 0.4417 & 0.6231 & 124 \\
            EXP-0.05 & 2.1080 & 0.5411 & 0.3662 & 0.5155 & 277 \\
            EXP-0.10 & 2.0547 & 0.5472 & 0.3618 & 0.5091 & 284 \\
            EXP-0.50 & 2.0014 & 0.5594 & 0.3583 & 0.5029 & 202 \\
            EXP-1.00 & 1.6655 & 0.5570 & 0.3268 & 0.4583 & 298 \\
            EXP-1.50 & 1.8670 & 0.5635 & 0.3459 & 0.4851 & 245 \\
            \bottomrule
        \end{tabularx}
        \caption{Out-of-sample results (with transaction fees) on Russell 2000 from 2005 to 2020  (using SR factor).}\label{tab:russell2000_tranfee_factor_SR}
    \end{table}

    \begin{table}[!ht]\small
            \begin{tabularx}{\textwidth}{lPPPPP}
                \toprule
                \textbf{Method} & \textbf{Return} &  \textbf{Maximum Drawdown} & \textbf{Sharpe Ratio} & \textbf{Sortino Ratio} & \textbf{Avg. Num. of Assets}\\
                \midrule
                Benchmark &&&\\
                \hline
                \hline\addlinespace[0.2em]
                EW & 3.1023 & 0.6125 & 0.4176 & 0.5948 & 1640 \\
                GMV-P & -0.6222 & 0.8962 & -0.0267 & -0.0394 & 674 \\
                GMV-LS & 0.0431 & 0.6830 & 0.1108 & 0.1656 & 768 \\
                GMV-NLS & -0.0241 & 0.6908 & 0.0961 & 0.1434 & 827 \\
                MV-P & 0.9973 & 0.5862 & 0.2644 & 0.4067 & 104 \\
                MV-LS & 2.2465 & 0.5843 & 0.3473 & 0.5363 & 29 \\
                MV-NLS & 1.2464 & 0.6185 & 0.2870 & 0.4370 & 2 \\ 
                \midrule
                Our methods &&&\\
                \hline
                \hline\addlinespace[0.2em]
                LOG & 3.1513 & 0.5204 & 0.4442 & 0.6253 & 109 \\
                EXP-0.05 & 3.0830 & 0.5074 & 0.4378 & 0.6147 & 432 \\
                EXP-0.10 & 2.6984 & 0.5160 & 0.4112 & 0.5761 & 417 \\
                EXP-0.50 & 2.8171 & 0.5148 & 0.4194 & 0.5879 & 377 \\
                EXP-1.00 & 3.1870 & 0.4988 & 0.4470 & 0.6274 & 305 \\
                EXP-1.50 & 3.1509 & 0.5000 & 0.4453 & 0.6251 & 207 \\
                \bottomrule
            \end{tabularx}
            \caption{Out-of-sample results (without transaction fees) on Russell 2000 from 2005 to 2020 (using RSI factor).}\label{tab:russell2000_factor_RSI}
        \end{table}

    \begin{table}[!h]\small
        \begin{tabularx}{\textwidth}{lPPPPP}
            \toprule
            \textbf{Method} & \textbf{Return} &  \textbf{Maximum Drawdown} & \textbf{Sharpe Ratio} & \textbf{Sortino Ratio} & \textbf{Avg. Num. of Assets}\\
            \midrule
            Benchmark &&&\\
            \hline
            \hline\addlinespace[0.2em]
            EW & 0.3844 & 0.6544 & 0.1665 & 0.2346 & 1640 \\
            GMV-P & -0.8261 & 0.9501 & -0.1614 & -0.2365 & 674 \\
            GMV-LS & -0.5611 & 0.8383 & -0.0745 & -0.1104 & 768 \\
            GMV-NLS & -0.6353 & 0.8589 & -0.1147 & -0.1696 & 827 \\
            MV-P & 0.5840 & 0.5977 & 0.2251 & 0.3457 & 104 \\
            MV-LS & 1.7918 & 0.5968 & 0.3221 & 0.4969 & 29 \\
            MV-NLS & 0.9993 & 0.6263 & 0.2676 & 0.4072 & 2 \\
            \midrule
            Our methods &&&\\
            \hline
            \hline\addlinespace[0.2em]
            LOG & 2.3978 & 0.5271 & 0.3912 & 0.5496 & 109 \\
            EXP-0.05 & 1.3793 & 0.5401 & 0.2962 & 0.4134 & 432 \\
            EXP-0.10 & 1.2245 & 0.5355 & 0.2785 & 0.3881 & 417 \\
            EXP-0.50 & 1.3409 & 0.5292 & 0.2917 & 0.4066 & 377 \\
            EXP-1.00 & 1.8099 & 0.5102 & 0.3411 & 0.4768 & 305 \\
            EXP-1.50 & 2.0923 & 0.5120 & 0.3669 & 0.5136 & 207 \\
            \bottomrule
        \end{tabularx}
        \caption{Out-of-sample results (with transaction fees) on Russell 2000 from 2005 to 2020 (using RSI factor).}\label{tab:russell2000_tranfee_factor_RSI}
    \end{table}

    To sum up, the proposed methods outperform the benchmarks with and without the transaction fees by incorporating factor signals. The performance of MV-type and GMV-type methods are not stable in the long run.
\end{document}